\documentclass[envcountsect,envcountsame,oribibl,orivec]{llncs}
\usepackage{amssymb}
\usepackage{amsmath}
\usepackage{url}
\usepackage{enumitem}
\usepackage{colonequals}
\usepackage{MnSymbol}

\usepackage{graphics}
\usepackage{graphicx}
\usepackage{hyperref}

\makeatletter
\newcommand{\todo}[1]{\marginpar{\textbf{TODO\footnotemark}}\@latex@warning{TODO: #1}\footnotetext{ #1}}
\makeatother
\usepackage{etoolbox}
\makeatletter
\let\llncs@addcontentsline\addcontentsline
\patchcmd{\maketitle}{\addcontentsline}{\llncs@addcontentsline}{}{}
\patchcmd{\maketitle}{\addcontentsline}{\llncs@addcontentsline}{}{}
\patchcmd{\maketitle}{\addcontentsline}{\llncs@addcontentsline}{}{}
\makeatother
\usepackage{hyperref}
\usepackage{bookmark}
\usepackage{tikz}
\newcommand*\circled[1]{\tikz[baseline=(char.base)]{
		\node[shape=circle,draw,inner sep=0.4pt] (char) {#1};}}

\newcommand{\Prop}{\textsf{Prop}}
\newcommand{\Formulae}{\textsf{Fml}}

\newcommand{\lfalse}{\bot}
\newcommand{\lneg}{\neg}

\newcommand{\propax}{\ensuremath{(\textsf{Taut})}}
\newcommand{\lrule}[2]{\displaystyle{\frac{#1}{#2}}}
\newcommand{\mprule}{\ensuremath{(\textsf{MP})}}
\newcommand{\limplies}{\rightarrow}
\newcommand{\liff}{\leftrightarrow}

\newcommand{\lnext}{\bigcirc}
\newcommand{\lalways}{\Box}
\newcommand{\leventually}{\Diamond}
\newcommand{\lonce}{\diamondminus}
\newcommand{\lsofar}{\boxminus}
\newcommand{\luntil}{{\,\mathcal{U}\,}}
\newcommand{\lsince}{{\,\mathcal{S}\,}}
\newcommand{\ltime}{{\,\texttt{time}\,}}
\newcommand{\ltrue}{{\,\texttt{true}\,}}

\newcommand{\wprevious}{\circled{\textit{w}}}
\newcommand{\sprevious}{\circled{\textit{s}}}

\newcommand{\kax}{\ensuremath{\textsf{-k}}}
\newcommand{\nextkax}{\ensuremath{(\lnext\kax)}}
\newcommand{\alwayskax}{\ensuremath{(\lalways\kax)}}
\newcommand{\funax}{\ensuremath{(\textsf{fun})}}

\newcommand{\mixaxone}{\ensuremath{(\textsf{mix1})}}
\newcommand{\mixaxtwo}{\ensuremath{(\textsf{mix2})}}
\newcommand{\indax}{\ensuremath{(\textsf{ind})}}
\newcommand{\uoneax}{\ensuremath{(\luntil\textsf{1})}}
\newcommand{\utwoax}{\ensuremath{(\luntil\textsf{2})}}
\newcommand{\necrule}{\ensuremath{\textsf{-nec}}}
\newcommand{\nextnecrule}{\ensuremath{(\lnext\necrule)}}
\newcommand{\alwaysnecrule}{\ensuremath{(\lalways\necrule)}}
\newcommand{\prevnecrule}{\ensuremath{(\wprevious\necrule)}}
\newcommand{\sofarnecrule}{\ensuremath{(\lsofar\necrule)}}
\newcommand{\uindrule}{\ensuremath{(\luntil\textsf{-R})}}
\newcommand{\sindrule}{\ensuremath{(\lsince\textsf{-R})}}
\newcommand{\prevkax}{\ensuremath{(\wprevious\kax)}}
\newcommand{\sofarkax}{\ensuremath{(\lsofar\kax)}}
\newcommand{\swprevax}{\ensuremath{(\textsf{sw})}}
\newcommand{\fpax}{\ensuremath{(\textsf{FP})}}
\newcommand{\pfax}{\ensuremath{(\textsf{PF})}}
\newcommand{\initialax}{\ensuremath{(\textsf{initial})}}
\newcommand{\sofarindax}{\ensuremath{(\lsofar\textsf{-ind})}}
\newcommand{\soneax}{\ensuremath{(\lsince\textsf{1})}}
\newcommand{\stwoax}{\ensuremath{(\lsince\textsf{2})}}

\newcommand{\LTL}{\textsf{LTL}}
\newcommand{\LTLp}{\LTL^{\sf P}}
\newcommand{\LPLTLp}{\textsf{LPLTL}^{\sf P}}
\newcommand{\LPLTL}{\textsf{LPLTL}}

\newcommand{\JLTL}{\textsf{LTL}^{\sf J}}

\newcommand{\lknows}{\mathsf{K}}

\newcommand{\SFive}{\textsf{S5}}

\newcommand{\nlax}{\ensuremath{(\textsf{nl})}}

\newcommand{\jnlax}{\ensuremath{(\textsf{jnl})}}
\newcommand{\jnfax}{\ensuremath{(\textsf{jnf})}}

\newcommand{\nfax}{\ensuremath{(\textsf{nf})}}

\newcommand{\CTerms}{\textsf{Const}}
\newcommand{\VTerms}{\textsf{Var}}
\newcommand{\Terms}{\textsf{Tm}}
\newcommand{\Ag}{\textsf{Ag}}

\newcommand{\jbox}[1]{\left[#1\right]\!}
\newcommand{\tapp}{\cdot}
\newcommand{\tinspect}{!}

\newcommand{\tnext}{\Rrightarrow}
\newcommand{\tprev}{\Lleftarrow}
\newcommand{\talwaysaccess}{\Downarrow}
\newcommand{\tgeneralize}{\Uparrow}
\newcommand{\tnextaccess}{\downarrow}

\newcommand{\thenceforthaccess}{\Downarrow_P}
\newcommand{\thenceforthgeneralize}{\Uparrow_P}
\newcommand{\tprevaccess}{\downarrow_P}
\newcommand{\tsprevright}{\Rrightarrow_P}
\newcommand{\tsprevleft}{\Lleftarrow_P}
\newcommand{\twprevright}{\Rightarrow_P}

\newcommand{\LP}{\textsf{LP}}

\newcommand{\appax}{\ensuremath{(\textsf{application})}}
\newcommand{\sumax}{\ensuremath{(\textsf{sum})}}
\newcommand{\posintax}{\ensuremath{(\textsf{positive introspection})}}

\newcommand{\refax}{\ensuremath{(\textsf{reflexivity})}}
\newcommand{\constnecrule}{\ensuremath{(\textsf{ax}\necrule)}}
\newcommand{\fpappax}{\ensuremath{(\textsf{FP-application})}}
\newcommand{\fpsumax}{\ensuremath{(\textsf{FP-sum})}}
\newcommand{\fpposintax}{\ensuremath{(\textsf{FP-positive introspection})}}

\newcommand{\iteratedconstnecrule}{\ensuremath{(\textsf{iax}\necrule)}}

\newcommand{\CS}{\textsf{CS}}

\newcommand{\numberofagents}{h}
\newcommand{\agent}{i}

\newcommand{\alwaysaccessprinciple}{\ensuremath{(\lalways\textsf{-access})}}
\newcommand{\generalizeprinciple}{\ensuremath{(\textsf{generalize})}}
\newcommand{\nextaccessprinciple}{\ensuremath{(\lnext\textsf{-access})}}
\newcommand{\nextrightshiftprinciple}{\ensuremath{(\lnext\textsf{-right})}}
\newcommand{\nextleftshiftprinciple}{\ensuremath{(\lnext\textsf{-left})}}

\newcommand{\pastgeneralizeprinciple}{\ensuremath{(\lsofar\textsf{-generalize})}}
\newcommand{\pastaccessprinciple}{\ensuremath{(\lsofar\textsf{-access})}}
\newcommand{\wprevaccessprinciple}{\ensuremath{(\wprevious\textsf{-access})}}
\newcommand{\wprevrightshiftprinciple}{\ensuremath{(\wprevious\textsf{-right})}}
\newcommand{\sprevrightshiftprinciple}{\ensuremath{(\sprevious\textsf{-right})}}
\newcommand{\sprevleftshiftprinciple}{\ensuremath{(\sprevious\textsf{-left})}}

\newcommand{\generalizeevidence}{\ensuremath{(generalize\text{-}\evidence)}}
\newcommand{\alwaysaccessevidence}{\ensuremath{(\lalways\text{-access-}\evidence)}}
\newcommand{\nextaccessevidence}{\ensuremath{(\lnext\text{-access-}\evidence)}}
\newcommand{\nextrightshiftevidence}{\ensuremath{(\lnext\text{-right-}\evidence)}}
\newcommand{\nextleftshiftevidence}{\ensuremath{(\lnext\text{-left-}\evidence)}}

\newcommand{\pastgeneralizeevidence}{\ensuremath{(\lsofar\text{-generalize-}\evidence)}}
\newcommand{\pastaccessevidence}{\ensuremath{(\lsofar\text{-access-}\evidence)}}
\newcommand{\wprevaccessevidence}{\ensuremath{(\wprevious\text{-access-}\evidence)}}
\newcommand{\wprevrightshiftevidence}{\ensuremath{(\wprevious\text{-right-}\evidence)}}
\newcommand{\sprevrightshiftevidence}{\ensuremath{(\sprevious\text{-right-}\evidence)}}
\newcommand{\sprevleftshiftevidence}{\ensuremath{(\sprevious\text{-left-}\evidence)}}

\newcommand{\jnfevidence}{\ensuremath{(\textsf{jnf-}\evidence)}}
\newcommand{\jnlevidence}{\ensuremath{(\textsf{jnl-}\evidence)}}

\newcommand{\nextleftshiftR}{\ensuremath{(\lnext\text{-left-R})}}
\newcommand{\localstates}{\mathbf{L}}
\newcommand{\runs}{\mathcal{R}}
\newcommand{\system}{\mathcal{I}}
\newcommand{\accrel}{\sim}
\newcommand{\evidence}{\mathcal{E}}
\newcommand{\valuation}{\nu}
\newcommand{\entails}{\vDash}
\newcommand{\proves}{\vdash}

\newcommand{\M}{\mathcal{M}}

\newcommand{\N}{\mathbb{N}}
\newcounter{enumsave}
\renewcommand{\phi}{\varphi}

\newcommand{\Sub}{\mathsf{Sub}}
\newcommand{\Subf}{\mathsf{Subf}}
\newcommand{\MCS}{\mathsf{MCS}}

\newcommand{\B}{\mathcal{B}}
\newcommand{\Op}[1]{\Phi^{#1}}
\newcommand{\OpB}{\Op{\B}_i}

\newcommand{\RO}[4]{#1 R_\lnext #2\ [#3, #4]}

\begin{document}

\title{Linear Temporal Justification Logics with Past Operators}
\author{Meghdad Ghari\thanks{This research was in part supported by a grant from IPM (No. 96030426).}
}
\institute{Department of Philosophy, Faculty of Literature and Humanities,\\
	University of Isfahan, Isfahan, Iran \\ and \\ School of Mathematics,
	Institute for Research in Fundamental Sciences (IPM), \\ P.O.Box: 19395-5746, Tehran, Iran \\ \email{ghari@ipm.ir}
}

\maketitle

\begin{abstract}
In this paper we present various temporal justification logics involving both past and future time modalities. We combine Artemov's logic of proofs with linear temporal logic (with both past and future operators), and establish its soundness and completeness. Then we investigate several principles describing the interaction of justification and time.
\end{abstract}

 \section{Introduction}
 \label{sec:Introduction}

Linear temporal logics of knowledge are useful for reasoning about situations where the knowledge of an agent  is changed over time \cite{FHMV95,HvdMV04,HZ92}. The temporal component in such systems is usually  interpreted over a discrete linear model of time with finite past and infinite future; in this case $(\N,<)$ can be chosen as the flow of time (for a logic of knowledge and branching time see \cite{vdMW03}). And the knowledge component is typically modeled using the modal logic $\SFive$.

This paper continues the study of temporal justification logics from \cite{Bucheli15,BucheliGhariStuder2017}. Temporal justification logic is a new family of temporal logics of knowledge in which the knowledge of agents is modeled using a justification logic. Justification logics are modal-like logics that provide a framework for reasoning about epistemic justifications (see \cite{Art08RSL,ArtFit11SEP} for a survey). The language of multi-agent justification logics extends the language of propositional logic by justification terms and expressions of the form $\jbox{t}_i \phi$, with the intended meaning ``$t$   is agent $i$'s justification for $\phi$.''  The \emph{Logic of Proofs} $\LP$ was the first logic in the family of justification logics, introduced by  Artemov in \cite{Art95TR,Art01BSL}. The logic of proofs is a justification counterpart of the modal epistemic logic {\sf S4}.

It is known that linear temporal logic with only future time operators is weak to fully express some properties of systems, such as \textit{unique initial states} and \textit{synchrony} (cf. \cite{HvdMV04,FrenchMeydenReynolds2004}). Neither  of the temporal justification logics of \cite{Bucheli15} and \cite{BucheliGhariStuder2017} contains past time operators in their languages. The aim of this paper is to add past time operators to the temporal justification logic of \cite{BucheliGhariStuder2017}, and to study principles describing the interaction of justifications and time.

 \section{Language}
 \label{sec:Syntax}

  In the following, let $\numberofagents$ be a fixed number of agents, $\Ag = \{1, \ldots, \numberofagents \}$ the set of all agents, $\CTerms$ a countable set of justification constants, $\VTerms$ a countable set of justification variables, and $\Prop$ a countable set of atomic propositions.

 The set of justification terms $\Terms$ is defined inductively by
 \[
  t \coloncolonequals c \mid x \mid \; \tinspect t \mid \;  t + t \mid t \tapp t \, ,
 \]
 where $c \in \CTerms$ and $x \in \VTerms$.

 The set of formulas $\Formulae$ is inductively defined by 
 \[
  \phi \coloncolonequals P \mid \lfalse \mid \phi \limplies \phi \mid \lnext \phi \mid \wprevious \phi \mid \phi \luntil \phi \mid \phi \lsince \phi \mid \jbox{t}_\agent\phi \, , 
 \]
 where $i \in \Ag$, $t \in \Terms$,  and $P \in \Prop$. The temporal operators $\lnext,\wprevious, \luntil,\lsince$ are respectively called \textit{next (or tomorrow), weak previous (or weak yesterday), until}, and \textit{since}. An \textit{until formula} is a formula of the form $\phi \luntil \psi$ for some formulas $\phi$ and $\psi$, and a \textit{justification assertion} is a formula of the form $\jbox{t}_i \phi$ for some formula $\phi$ and term $t$.

We use the following usual abbreviations:
 \begin{align*}
  \lneg \phi &\colonequals \phi \limplies \lfalse &
  \top &\colonequals \lneg \lfalse &\\
  \phi \lor \psi &\colonequals \lneg \phi \limplies \psi &
  \phi \land \psi &\colonequals \lneg (\lneg \phi \lor \lneg \psi) \\
  \phi \liff \psi &\colonequals (\phi \limplies \psi) \land (\psi \limplies \phi) & \sprevious \phi &\colonequals \neg \wprevious \neg \phi
   \\
  \leventually \phi &\colonequals \top \luntil \phi & \lalways \phi &\colonequals \lneg \leventually \lneg \phi  \\
 \lonce \phi &\colonequals \top \lsince \phi  & \lsofar \phi &\colonequals \lneg \lonce\lneg \phi.
 \end{align*}

The temporal operators $\sprevious, \lalways, \leventually, \lsofar, \lonce$ are respectively called \textit{strong previous, always from now on (or henceforth), sometime (or eventuality), has-always-been}, and \textit{once}.

Associativity and precedence of connectives, as well as the corresponding omission of brackets, are handled in the usual manner.

Subformulas are defined as usual. The set of subformulas $\Sub(\chi)$ of a formula $\chi$ is inductively given by:
 \begin{align*}
 \Sub(P) &\colonequals \{P\}  & \Sub(\bot) &\colonequals \{\bot\}\\
 \Sub(\phi \to \psi) &\colonequals \{\phi \to \psi\} \cup \Sub(\phi) \cup \Sub(\psi) & \Sub(\jbox{t}_\agent\phi) &\colonequals \{\jbox{t}_\agent\phi\} \cup \Sub(\phi) \\
  \Sub(\phi \luntil \psi ) &\colonequals \{\phi \luntil \psi\} \cup \Sub(\phi) \cup \Sub(\psi) & \Sub(\lnext \phi) &\colonequals \{\lnext \phi\} \cup \Sub(\phi) \\
  \Sub(\phi \lsince \psi ) &\colonequals \{\phi \lsince \psi\} \cup \Sub(\phi) \cup \Sub(\psi) & \Sub(\wprevious\phi) &\colonequals \{\wprevious\phi\} \cup \Sub(\phi).
 \end{align*}

For a set $S$ of formulas, $\Sub(S)$ denotes the set of all subformulas of
the formulas from $S$.

The combined language of justification logic and temporal logic allows for expressing some properties of systems that are not expressible in the known logics of knowledge and time. For example,

\begin{itemize}
	\item ``$t$ justifies $\phi$ for agent $i$ until $\psi$ holds" can be expressed by $(\jbox{t}_i \phi) \luntil \psi$.  

	\item ``$t$ justifies $\phi$ for agent $i$ since $\psi$ holds" can be expressed by $(\jbox{t}_i \phi) \lsince \psi$. 
	
	\item ``$t$ is agent $i$'s \textit{conclusive evidence} that $\phi$ is true" can be expressed by $\lalways \jbox{t}_i \phi$ or even by $\lsofar \jbox{t}_i \phi \wedge \lalways  \jbox{t}_i \phi$.
	
	\item ``If agent $i$ knows that $\phi$ for reason $t$, then  $\phi$ is always true" can be expressed by $\jbox{t}_i \phi \to \lalways \phi$. 
	
	\item ``Agent $i$ will have not forgotten her justification $t$ for $\phi$ by tomorrow, providing she possesses the justification now" can be expressed by $\jbox{t}_i \phi \to \lnext \jbox{t}_i  \phi$. 
	
	\item ``Agent $i$ will learn that $t$ is a justification for $\phi$  tomorrow, but she does not know it now" can be expressed by $\neg \jbox{t}_i \phi \wedge \lnext \jbox{t}_i  \phi$. 
\end{itemize}

More connections between justification and time will be explored in Sections \ref{sec:Connecting principles}, \ref{sec:No forgetting and no learning}, and \ref{sec:JLTL}.
 \section{Axioms}
 \label{sec:Axioms}

  The axiom system for temporal justification logic consists of three parts, namely propositional logic, temporal logic, and justification logic.

 \subsection*{Propositional Logic}
 For propositional logic, we take
 \begin{enumerate}
  \setcounter{enumi}{\theenumsave}
  \item all propositional tautologies \hfill \propax
  \setcounter{enumsave}{\theenumi}
 \end{enumerate}
 as axioms and the rule modus ponens, as usual:
 \[
   \lrule{\vdash \phi \qquad \vdash \phi \limplies \psi}{\vdash \psi}\,\mprule \, .
 \]

 \subsection*{Temporal Logic}
 For the temporal part, we use a system of~\cite{Gabbay,Gol87,Gor99} and  \cite{PnueliLichtensteinZuck1985,PnueliLichtenstein2000,FrenchMeydenReynolds2004}
 with axioms \\
 
\noindent {\bf Axioms for the future operators:}
 \begin{enumerate}
  \setcounter{enumi}{\theenumsave}
  \item $\lnext( \phi \limplies \psi) \limplies (\lnext \phi \limplies \lnext \psi)$ \hfill \nextkax
  \item $\lalways( \phi \limplies \psi) \limplies (\lalways \phi \limplies \lalways \psi)$ \hfill \alwayskax
  \item $\lnext \lneg \phi \liff \lneg \lnext \phi$ \hfill \funax
  \item $\lalways (\phi \limplies \lnext \phi) \limplies (\phi \limplies \lalways \phi)$ \hfill \indax
  \item $\phi \luntil \psi \limplies \leventually \psi$ \hfill \uoneax
  \item $\phi \luntil \psi \liff \psi \lor (\phi \land \lnext(\phi \luntil \psi))$ \hfill \utwoax
  \setcounter{enumsave}{\theenumi}
\end{enumerate}
\noindent {\bf Axioms for the past operators:}
  \begin{enumerate}
  	\setcounter{enumi}{\theenumsave}
  \item $\lsofar( \phi \limplies \psi) \limplies (\lsofar \phi \limplies \lsofar \psi)$ \hfill \sofarkax
  \item $\wprevious( \phi \limplies \psi) \limplies (\wprevious \phi \limplies \wprevious \psi)$ \hfill \prevkax
  \item $\sprevious \phi \rightarrow \wprevious\phi$ \hfill \swprevax
  \item $\lonce \wprevious \bot$ \hfill \initialax
    \item $\lsofar (\phi \limplies \wprevious \phi) \limplies (\phi \limplies \lsofar \phi)$ \hfill \sofarindax
  \item $\phi \lsince \psi \rightarrow \lonce \psi$ \hfill \soneax
  \item $\phi \lsince \psi \liff \psi \lor (\phi \land \sprevious(\phi \lsince \psi))$ \hfill \stwoax
  \setcounter{enumsave}{\theenumi}
 \end{enumerate}
\noindent {\bf Axioms for the interaction of the future and past operators:}
\begin{enumerate}
	\setcounter{enumi}{\theenumsave}

	\item $\phi \rightarrow \lnext \sprevious \phi$ \hfill \fpax
	\item $\phi \rightarrow \wprevious \lnext \phi$ \hfill \pfax

	\setcounter{enumsave}{\theenumi}
\end{enumerate}
 and rules
 \[
  \lrule{\vdash \phi}{\vdash \lnext \phi}\,\nextnecrule \, , \qquad\lrule{\vdash \phi}{\vdash \wprevious \phi}\, \prevnecrule \, , \qquad \lrule{\vdash \phi}{\vdash \lalways\phi}\,\alwaysnecrule \, ,\qquad \lrule{\vdash \phi}{\vdash \lsofar\phi}\,\sofarnecrule.
 \]

Let $\LTLp$ denote the axiomatic system given by the above axioms and rules.
 \subsection*{Justification Logic}
 Finally, for the justification logic part, we use a multi-agent version of the Logic of Proofs~\cite{Art01BSL,BucKuzStu11JANCL,Ghari-TCS-2014,TYav08TOCS} with axioms 
 \begin{enumerate}
 	\setcounter{enumi}{\theenumsave}
 	\item $\jbox{t}_\agent (\phi \limplies \psi) \limplies (\jbox{s}_\agent \phi \limplies \jbox{t \tapp s}_\agent \psi)$ \hfill \appax
 	\item $\jbox{t}_\agent \phi \rightarrow \jbox{t + s}_\agent \phi$, $\jbox{s}_\agent \phi \limplies  \jbox{t + s}_\agent \phi$ \hfill \sumax
 	\item $\jbox{t}_\agent \phi \limplies \phi$ \hfill \refax
 	\item $\jbox{t}_\agent \phi \limplies \jbox{\tinspect t}_\agent \jbox{t}_\agent \phi$ \hfill \posintax
 	\setcounter{enumsave}{\theenumi}
 \end{enumerate}

and the iterated axiom necessitation rule
\[
\lrule{\jbox{c_{j_n}}_{i_n}\ldots\jbox{c_{j_1}}_{i_1} \phi \in \CS}{\vdash \jbox{c_{j_n}}_{i_n}\ldots\jbox{c_{j_1}}_{i_1} \phi}\ \iteratedconstnecrule
\]
where the \textit{constant specification} $\CS$ is a set of formulas of the form
 $$\jbox{c_{j_n}}_{i_n}\ldots\jbox{c_{j_1}}_{i_1} \phi$$ 
 where $n \geq 1$, $i_1,\ldots,i_n$ are arbitrary agents,  $c_{j_n},\ldots,c_{j_1}$ are justification constants, and  $\phi$ is an axiom instance of  propositional logic, temporal logic, or justification logic. Moreover, a constant specification $\CS$ should be downward closed in the sense that whenever $\jbox{c_{j_n}}_{i_n}\jbox{c_{j_{n-1}}}_{i_{n-1}}\ldots\jbox{c_{j_1}}_{i_1} \phi \in \CS$, then $\jbox{c_{j_{n-1}}}_{i_{n-1}}\ldots\jbox{c_{j_1}}_{i_1} \phi \in \CS$ for $n >1$.

\begin{definition}
	A constant specification $\CS$ for a justification logic $\mathsf{L}$ is \textit{axiomatically appropriate} provided, for every axiom instance $\phi$ of $\mathsf{L}$ and for every $n \geq 1$, and every $i_1, \ldots, i_n \in \Ag$,   $\jbox{c_{j_n}}_{i_n}\ldots\jbox{c_{j_1}}_{i_1} \phi \in \CS$ for some justification constants $c_{j_n},\ldots,c_{j_1}$.
\end{definition}

\begin{remark}
	It is perhaps worth noting that the temporal justification logics of \cite{Bucheli15, BucheliGhariStuder2017} are formalized using the following axiom necessitation rule
	 \[
	   \lrule{\jbox{c}_\agent \phi \in \CS}{\vdash \jbox{c}_\agent \phi}\, \constnecrule \, .
	 \]
	 We prefer $\iteratedconstnecrule$ to $\constnecrule$ because the iterated axiom necessitation rule enables us to prove the internalization property (see Section \ref{sec:Internalization}). All the results of this paper, except the results of Section \ref{sec:Internalization}, continue to hold if the logics are formalized by the rule $\constnecrule$.
\end{remark}

For a given constant specification $\CS$, we use $\LPLTLp_\CS$ to denote the Hilbert system given by the axioms and rules for propositional logic, temporal logic, and justification logic as presented above.
We write $\vdash_\CS \phi$  if a formula $\phi$ is derivable in $\LPLTLp_\CS$.
 
The definition of derivation from a set of premises is standard. A formula $\phi$ is derivable from the set of assumptions $\Gamma$, written $\Gamma \vdash_\CS \phi$, iff $\phi$ is in $\Gamma$, or is one of the axioms of $\LPLTLp_\CS$, or follows from derivable formulas through applications of the rules $\mprule$, $\constnecrule$, and necessitation, where  necessitation rules can be applied only to derivations without assumptions. In other words: 

\[ \frac{\phi \in \Gamma}{\Gamma \vdash_\CS \phi} \, , \qquad
\frac{\phi \in Axiom}{\Gamma \vdash_\CS \phi} \, , \qquad
\frac{\Gamma \vdash_\CS \chi~~~\Delta \vdash_\CS \chi \rightarrow \psi}{\Gamma,\Delta \vdash_\CS \psi} \, , \qquad \lrule{\jbox{c}_\agent \phi \in \CS}{\Gamma \vdash \jbox{c}_\agent \phi} \, ,\]

\[  \lrule{\vdash_\CS \phi}{\Gamma \vdash_\CS \lnext \phi} \, ,  \qquad\qquad \lrule{\vdash_\CS \phi}{\Gamma \vdash_\CS \wprevious \phi} \, , \qquad\qquad \lrule{\vdash_\CS \phi}{\Gamma \vdash_\CS \lalways\phi} \, , \qquad\qquad \lrule{\vdash_\CS \phi}{\Gamma \vdash_\CS \lsofar\phi} \,.
\]

Note that the Deduction Theorem holds in $\LPLTLp_\CS$. It is easy to show that:
\[
\Gamma \vdash_\CS \phi \text{ iff there exist $\psi_1,\ldots,\psi_n \in \Gamma$ such that 
	$\vdash_\CS (\psi_1 \land \cdots \land \psi_n) \to \phi$.}
\]

Temporal justification logic $\LPLTL$ of \cite{BucheliGhariStuder2017} is a fragment of $\LPLTLp$ without past operators $\wprevious$ and $\lsince$, and without axioms and rules involving past operators.

The axiomatization for linear time temporal logic given in~\cite{Gabbay,Gol87,Gor99,KrogerMerz2008} includes the following axioms 
\[
\lalways \phi \to (\phi \land \lnext \lalways \phi),
\]
\[
\lsofar \phi \to (\phi \wedge \wprevious \lsofar \phi).
\]
The following lemma shows that we do not need these axioms since in our formalization $\lalways$ and $\lsofar$ are defined operators. 

\begin{lemma}\label{lem:mix:1}
The following formulas are provable in $\LTLp$:

\begin{enumerate}
	\item $\lalways \phi \to (\phi \land \lnext \lalways \phi)$.
	\item $\lalways \phi \to  \lnext \phi$.
	\item $\lsofar \phi \to (\phi \wedge \wprevious \lsofar \phi)$.
	\item $\lsofar \phi \to  \wprevious \phi$.
\end{enumerate}

In item 1, $\mprule$ is the only rule that is used in the derivation.
\end{lemma}
\begin{proof}
\begin{enumerate}
	\item 
$\lalways \phi$ stands for $\lnot ( \top \luntil \lnot\phi)$.
Hence from $\utwoax$ we get
\[
  \lnot \phi \lor \lnext (\top \luntil \lnot \phi) \to \top \luntil \lnot \phi.
\]
Taking the contrapositive yields
\[
  \lnot ( \top \luntil \lnot \phi) \to \lnot (\lnot \phi \lor \lnext (\top \luntil \lnot \phi)).
\]
By propositional reasoning and $\funax$ we get
\[
 \lnot ( \top \luntil \lnot \phi) \to  (\phi \land \lnext \lnot (\top \luntil \lnot \phi)),
\]
which is
\[
  \lalways \phi \to (\phi \land \lnext \lalways \phi).
\]

\item From item 1 and propositional reasoning we get
\begin{align}
\lalways \phi \limplies \phi \label{eq:mixax1}\\
\lalways \phi \limplies \lnext\lalways \phi \label{eq:mixax2}
\end{align}
From~\eqref{eq:mixax1} and $\nextnecrule$ we get 
\[
\lnext(\lalways \phi \limplies \phi)
\]
which, in turn, using $\nextkax$ and propositional reasoning gives
\[
\lnext\lalways \phi \limplies \lnext \phi \, .
\]
By propositional reasoning and using~\eqref{eq:mixax2} we obtain 
\[
\lalways \phi \limplies \lnext \phi.
\]

\item Similar to item 1.

\item Similar to item 2.\qed
\end{enumerate}
\end{proof}

\begin{lemma}\label{lem: properties of previous operators}
	The following formulas are provable in $\LTLp$:
	\begin{enumerate}
		\item $\sprevious \phi \rightarrow \neg \wprevious \bot$.
		
		\item $\wprevious(\phi_1 \vee \ldots \vee \phi_n) \leftrightarrow (\wprevious \phi_1 \vee \ldots \vee \wprevious \phi_n)$.
		
		\item $\sprevious(\phi_1 \vee \ldots \vee \phi_n) \leftrightarrow (\sprevious \phi_1 \vee \ldots \vee \sprevious \phi_n)$.
		
		\item $\wprevious(\phi_1 \vee \phi_2 \vee \ldots \vee \phi_{n-1} \vee \phi_n) \leftrightarrow (\sprevious \phi_1 \vee \sprevious \phi_2 \vee\ldots \vee \sprevious \phi_{n-1} \vee \wprevious \phi_n)$.
		
		\item $\wprevious (\phi_1 \wedge \ldots \wedge \phi_n) \leftrightarrow (\wprevious \phi_1 \wedge \ldots \wedge \wprevious \phi_n)$.
		
		\item $\lnext (\phi_1 \wedge \ldots \wedge \phi_n) \leftrightarrow (\lnext \phi_1 \wedge \ldots \wedge \lnext \phi_n)$.
		
		\item  $\lnext (\phi_1 \vee \ldots \vee \phi_n) \leftrightarrow (\lnext \phi_1 \vee  \ldots \vee \lnext \phi_n)$.
	\end{enumerate}

\end{lemma}

\begin{lemma}\label{lem: relationship beetwin since and until}
	The following formulas are provable in $\LTLp$:
	\begin{enumerate}
		\item $\phi \wedge \psi \luntil \sigma \limplies \psi \luntil (\sigma \wedge (\sprevious \psi) \lsince \phi)$.
		
		\item $\phi \wedge \psi \lsince \sigma \limplies \psi \lsince (\sigma \wedge (\lnext \psi) \luntil \phi)$.
		
	\end{enumerate}
	
\end{lemma}

\begin{lemma}\label{lem: derivable ind-rules}
	The following rules are derivable in $\LTLp$:
\[
\lrule{\vdash \phi \to \psi \quad \vdash\phi \to \lnext \phi}{\vdash \phi \to \lalways \psi}
\qquad
\lrule{\vdash \phi \to \lnext \phi}{\vdash \phi \to \lalways \phi}
\]

\[
\lrule{\vdash \phi \to \psi \quad \vdash\phi \to \wprevious \phi}{\vdash \phi \to \lsofar \psi}
\qquad
\lrule{\vdash \phi \to \wprevious \phi}{\vdash \phi \to \lsofar \phi}
\]
	
\end{lemma}

\begin{lemma}
The following rules are derivable in $\LTLp$:
\[  
\lrule{\chi \limplies \lneg \psi \land \lnext \chi}{\chi \limplies \lneg(\phi \luntil \psi)}\, 
\qquad 
\lrule{\chi \limplies \lneg \psi \land \wprevious \chi}{\chi \limplies \lneg(\phi \lsince \psi)}\, 
\]
\[  
\lrule{\chi \limplies \lneg \psi \land \lnext (\chi \lor (\lnot \phi \land \lnot \psi)) }{\chi \limplies \lneg(\phi \luntil \psi)}\,\uindrule \,
 \qquad
 \lrule{\chi \limplies \lneg \psi \land \wprevious (\chi \lor (\lnot \phi \land \lnot \psi)) }{\chi \limplies \lneg(\phi \lsince \psi)}\,\sindrule \, . 
\]
\end{lemma}

\section{Maximal consistent sets}
\label{sec:Maximal consistent sets}

All the results of this section hold for extensions of $\LTLp$, i.e. $\LPLTL$, $\LPLTLp$, and all extensions introduced in Sections \ref{sec:Connecting principles}, \ref{sec:Internalization}, \ref{sec:JLTL}. Let $\mathsf{L}$ be an extension of $\LTLp$, and let $\vdash_\CS$ denote derivability in $\mathsf{L}_\CS$, where $\CS$ is a constant specification for ${\sf L}$.

For a formula $\chi$, let 
$$A_\chi := \Sub(\chi) \cup \Sub(\top \lsince \wprevious \bot) ,$$
$$\Sub^+(\chi) := A_\chi \cup \{ \neg \psi \ |\  \psi \in A_\chi \}.$$

\begin{definition}
	Let $\CS$ be a constant specification for $\mathsf{L}$.
	\begin{itemize}
		\item A set $\Gamma$ of formulas is called \emph{$\mathsf{L}_\CS$-consistent} (or simply $\CS$-consistent) if $\Gamma \not\vdash_\CS \bot$.
		
		\item A set $\Gamma$ of formulas is called \textit{maximal} if it has no $\mathsf{L}_\CS$-consistent proper extension of formulas.
		
		\item A set $\Gamma \subseteq \Sub^+(\chi)$  is called \textit{$\chi$-maximal} if it has no $\mathsf{L}_\CS$-consistent proper extension of formulas from $\Sub^+(\chi)$.
	\end{itemize}

\end{definition} 

Let $\MCS_\chi$ denote the set of all $\chi$-maximally $\mathsf{L}_\CS$-consistent subsets of $\Sub^+(\chi)$. Note that $\MCS_\chi$ is a finite set.

Let $\MCS$ denote the set of all maximally $\mathsf{L}_\CS$-consistent sets, and for $\Gamma\in \MCS$, let $$\overline{\Gamma} := \Gamma \cap \Sub^+(\chi).$$ 

\begin{lemma}\label{lem:characterization of MCS-X}
\[
\MCS_\chi = \{\overline{\Gamma} \mid \Gamma\in\MCS \}.
\]
\end{lemma}
\begin{proof}
     $(\subseteq)$ Let $\Delta \in \MCS_\chi$. Then $\Delta$ can be extended to a maximal $\CS$-consistent set $\Gamma \in \MCS$. It is easy to show that $\Delta = \Gamma \cap \Sub^+(\chi)$,  and thus $\Delta = \overline{\Gamma}$.
     
     $(\supseteq)$ It is sufficient to show that for each $\Gamma \in \MCS$ the set $\overline{\Gamma}$ is $\CS$-consistent and $\chi$-maximal. The $\CS$-consistency of $\overline{\Gamma}$ follows from the $\CS$-consistency of $\Gamma$. In order to show the $\chi$-maximality of $\overline{\Gamma}$, suppose towards a contradiction that $\overline{\Gamma}$ has a $\CS$-consistent proper extension $\Sigma \subseteq \Sub^+(\chi)$. Let $\phi \in \Sigma \setminus \overline{\Gamma}$. Thus $\phi \not \in \Gamma$, and hence $\neg \phi \in \Gamma$. Since $\phi \in \Sub^+(\chi)$ we can distinguish the following cases:
          \begin{itemize}
     	\item $\phi \in A_\chi$. In this case $\neg \phi \in \Sub^+(\chi)$, and hence $\neg \phi \in \overline{\Gamma} \subseteq \Sigma$, which contradicts $\phi \in \Sigma$.
     	
     	\item $\phi = \neg \psi$ and $\psi \in A_\chi$. In this case $\neg \phi = \neg \neg \psi \in \Gamma$, and hence $\psi \in \Gamma$. Thus $\psi \in \overline{\Gamma} \subseteq \Sigma$, which contradicts $\neg \psi \in \Sigma$. \qed
     \end{itemize}
\end{proof}

\begin{lemma}\label{lem: Facts about MCS-chi}
	Let  $\overline{\Gamma} \in \MCS_\chi$.
	
	\begin{enumerate}
		\item If 
		$\overline{\Gamma} \vdash_\CS \phi$, then $\vdash_\CS \bigwedge \overline{\Gamma} \to \phi$.
		
		\item If $\phi \in A_\chi$ and  $\phi \not\in \overline{\Gamma}$, then $\neg \phi \in \overline{\Gamma}$.
		
		\item If $\phi \in \Sub^+(\chi)$ and  $\overline{\Gamma} \vdash_\CS \phi$, then $\phi \in \overline{\Gamma}$.
		
		\item If $\psi \in \Sub^+(\chi)$, $\phi \in \overline{\Gamma}$ and $\vdash_\CS \phi \rightarrow \psi$, then $\psi \in \overline{\Gamma}$.
		
	\end{enumerate}
\end{lemma}
\begin{proof}
	The proof of all items are standard. \qed 
\end{proof}

The following definition is inspired by the work of  Gabbay et al. \cite{Gabbay-Pnueli-Shelah-Stavi-1980}.

\begin{definition}
	The relation $R_\lnext$ on $\MCS_\chi$ is defined as follows:
	\[
	X R_\lnext Y 
	\text{ if{f} }
	\text{there exists $\Gamma, \Delta \in \MCS$ such that $X = \overline{\Gamma}$ and $Y = \overline{\Delta}$ and  }  \{ \phi \ |\  \lnext \phi \in \Gamma\} \subseteq \Delta.
	\]
	\textbf{Notation.}	The notation $\RO{X}{Y}{\overline{\Gamma}}{\overline{\Delta}}$ means that $X, Y \in \MCS_\chi$, $\Gamma, \Delta \in \MCS$, $X = \overline{\Gamma}$, $Y = \overline{\Delta}$ and $\{ \phi \ |\  \lnext \phi \in \Gamma\} \subseteq \Delta$. 
\end{definition}

Note that for $\Gamma, \Delta \in \MCS$ if $\{ \phi \ |\  \lnext \phi \in \Gamma\} \subseteq \Delta$, then $\overline{\Gamma} R_\lnext \overline{\Delta}$. Hence, if $\RO{X}{Y}{\overline{\Gamma}}{\overline{\Delta}}$, then $\overline{\Gamma} R_\lnext \overline{\Delta}$. In addition, it is easy to show that if $\RO{X}{Y}{\overline{\Gamma}}{\overline{\Delta}}$, then $\{ \phi \ |\  \lnext \phi \in \Gamma\} = \Delta$ (see item 1 of Lemma \ref{lem:R-next}).

From the above definition we immediately get the following lemma.

\begin{lemma}\label{lem:R-next is serial}
	The relation $R_\lnext$ is serial. 
	That is for each $X \in \MCS_\chi$, there exists $Y \in \MCS_\chi$ such that $X R_\lnext Y$.
\end{lemma}
\begin{proof}
	For $X \in \MCS_\chi$, by Lemma \ref{lem:characterization of MCS-X}, there exists $\Gamma \in \MCS$ such that $X =\overline{\Gamma}$. Let
	\[
	\Delta \colonequals \{ \phi \mid \lnext \phi \in \Gamma \}.
	\]
	We prove that $\Delta \in \MCS$, and so $\overline{\Gamma} R_\lnext \overline{\Delta}$ as desired. We first show that $\Delta$ is $\CS$-consistent. If $\Delta$ is not $\CS$-consistent, then
	\[
	\vdash_\CS \phi_1 \wedge \ldots \wedge \phi_n \rightarrow \bot
	\]
	for some $\lnext\phi_1, \ldots, \lnext\phi_n \in \Gamma$. Thus 
	\[
	\vdash_\CS \lnext\phi_1 \wedge \ldots \wedge \lnext\phi_n \rightarrow \lnext\bot.
	\]
	Hence $\lnext \bot \in \Gamma$. Since $\vdash_\CS \lnext \bot \to \bot$, we have 
	$\bot \in \Gamma$ which is a contradiction. 
	
	In order to show the maximality of $\Delta$, suppose towards a contradiction that $\Delta$ has a $\CS$-consistent proper extension $\Sigma$. Let $\phi \in \Sigma \setminus \Delta$. Thus $\phi \not \in \Delta$, and hence $\neg \lnext \phi \in \Gamma$. From the latter it follows that $\lnext \neg \phi \in \Gamma$, and hence $\neg \phi \in \Delta \subseteq \Sigma$ which is a contradiction.  
	\qed
\end{proof}

\begin{lemma}\label{lem:R-next}
	Let $\RO{X}{Y}{\overline{\Gamma}}{\overline{\Delta}}$.
	\begin{enumerate}
		\item  $\lnext \phi \in \Gamma$ if{f} $\phi \in \Delta$.
		
		\item $\phi \in \Gamma$ if{f} $\sprevious \phi \in \Delta$.
		
		\item $\phi \in \Gamma$ if{f} $\wprevious \phi \in \Delta$.
	\end{enumerate}
\end{lemma}
\begin{proof}
	\begin{enumerate}
		\item The proof of the only if direction follows from the definition of $R_\lnext$. For the if direction, suppose that $\phi \in \Delta$, and suppose towards a contradiction that $\lnext \phi \not \in \Gamma$. Thus $\neg \lnext \phi \in \Gamma$, and hence $\lnext \neg \phi \in \Gamma$. Since $\{ \phi \ |\  \lnext \phi \in \Gamma\} \subseteq \Delta$, we get  $\neg \phi \in \Delta$, which would contradict the assumption.
		
		\item If $\phi \in \Gamma$, then by the axiom $\fpax$ we get $\lnext \sprevious\phi \in \Gamma$, and hence by $\{ \phi \ |\  \lnext \phi \in \Gamma\} \subseteq \Delta$ we get $\sprevious \phi \in \Delta$. For the converse, suppose that $\sprevious \phi \in \Delta$. Then, by item 1, $\lnext \sprevious \phi \in \Gamma$. Assume to obtain a contradiction that  $\phi \not \in \Gamma$. Then $\neg \phi \in \Gamma$, and by the axiom $\fpax$ we get  $\lnext \sprevious \neg \phi \in \Gamma$. Since $\vdash_\CS \lnext \sprevious \neg \phi \rightarrow \neg \lnext \sprevious  \phi$, we arrive at a contradiction  $\neg \lnext \sprevious  \phi \in \Gamma$.
		
		\item The only if direction is obtained from item 2 and axiom $\swprevax$. For the converse suppose $\wprevious \phi \in \Delta$. By item 1, it follows that $\lnext \wprevious \phi \in \Gamma$. Assume to obtain a contradiction that $\phi \not \in \Gamma$. Then $\neg \phi \in \Gamma$, and hence $\lnext \sprevious \neg \phi \in \Gamma$. Since $\vdash_\CS \lnext \sprevious \neg \phi \rightarrow \neg \lnext \wprevious \phi$, we arrive at a contradiction $\neg \lnext \wprevious \phi \in \Gamma$.     \qed
	\end{enumerate}
\end{proof}

\begin{definition}
	Given $X \in \MCS_\chi$, $X$ is called \textit{initial} if $\wprevious \bot \in X$.
\end{definition}

\begin{lemma}\label{lem:initial state}
	Given $X \in \MCS_\chi$, $X$ is not initial if{f} there exists $Y \in \MCS_\chi$ such that $Y R_\lnext X$.
\end{lemma}
\begin{proof}
	Suppose that $X \in \MCS_\chi$ is not initial. There is $\Gamma \in \MCS$ such that $X = \overline{\Gamma}$. Let
	\[
	\Lambda \colonequals \{ \lnext\phi \mid  \phi \in \Gamma \}.
	\]
	We first prove that $\Lambda$ is $\CS$-consistent.	If $\Lambda$ is not $\CS$-consistent, then
	\[
	\vdash_\CS \lnext\phi_1 \wedge \ldots \wedge \lnext\phi_n \rightarrow \bot
	\]
	for some $\phi_1, \ldots, \phi_n \in \Gamma$. Using item 1 of Lemma \ref{lem: properties of previous operators}, we have
	\[
	\vdash_\CS \wprevious\lnext\phi_1 \wedge \ldots \wedge \wprevious\lnext\phi_n \rightarrow \wprevious\bot.
	\]
	By axiom $\pfax$
	\[
	\vdash_\CS  \phi_1 \wedge \ldots \wedge \phi_n \rightarrow \wprevious\bot.
	\]
	Hence $\wprevious \bot \in \Gamma$. Since $\wprevious \bot \in  \Subf^+(\chi)$, we have $\wprevious \bot \in \overline{\Gamma}$, contradicting our assumption that $\overline{\Gamma}$ is not initial. Thus $\Lambda$ is $\CS$-consistent, and it can be extended to a maximally $\CS$-consistent set $\Delta \in \MCS$. It is easy to show that $\overline{\Delta} R_\lnext \overline{\Gamma}$. Finally put $Y := \overline{\Delta}$.
	
	Conversely, suppose that there exists $Y \in \MCS_\chi$ such that $Y R_\lnext X$. \linebreak Let $\RO{Y}{X}{\overline{\Delta}}{\overline{\Gamma}}$. Suppose towards a contradiction that $\wprevious \bot \in \overline{\Gamma}$. By item 3 of Lemma \ref{lem:R-next}, we get $\bot \in \Delta$, which is a contradiction. \qed 
\end{proof}

\begin{lemma}\label{lem:initial state does not contain strong previous formulas}
	If $\overline{\Gamma}$ is initial, then  $\sprevious \phi \not \in \Gamma$ for all formulas $\phi$.
\end{lemma}
\begin{proof}
	Suppose $\overline{\Gamma}$ is initial. Then $\wprevious \bot \in \Gamma$. Now the result follows from the fact that  $\vdash_\CS   \wprevious \bot \limplies \neg \sprevious \phi$. \qed
\end{proof}

The following lemma provides another helpful characterization for the relation $R_\lnext$. In fact, this characterization is used as the definition  of $R_\lnext$ in \cite{HvdMV04}.

\begin{lemma}\label{lem:equivalence definition of R-next}
	Let $X,Y \in \MCS_\chi$. Then
	\[
	X R_\lnext Y
	\quad\text{ if{f} }\quad 
	\not\vdash_\CS \bigwedge X \to \lnot \lnext \bigwedge Y.
	\]
\end{lemma}
\begin{proof}
	$(\Rightarrow)$ Suppose that $X R_\lnext Y$, and thus $\RO{X}{Y}{\overline{\Gamma}}{\overline{\Delta}}$ for some $\Gamma, \Delta \in \MCS$. By item 1 of Lemma \ref{lem:R-next}, we get $\lnext \bigwedge \overline{\Delta} \in \Gamma$, and hence
	\( 
	\bigwedge \overline{\Gamma} \wedge \lnext \bigwedge \overline{\Delta} \in \Gamma.
	\)		
	Therefore, 
	\[
	\not \vdash_\CS \neg (\bigwedge \overline{\Gamma} \wedge \lnext \bigwedge \overline{\Delta}),
	\]
	and hence
	\[
	\nvdash_\CS \bigwedge \overline{\Gamma} \to \lnot \lnext \bigwedge \overline{\Delta}.
	\]

	$(\Leftarrow)$ Suppose $\nvdash_\CS \bigwedge X \to \lnot \lnext \bigwedge Y.$ Thus there is $\Gamma \in \MCS$ such that $\bigwedge X \wedge \lnext \bigwedge Y \in \Gamma$. Thus $X \subseteq \overline{\Gamma}$, which immediately implies $X = \overline{\Gamma}$, by the $\chi$-maximality of $X$. Now let
	\( 
	\Delta \colonequals \{ \phi \mid \lnext \phi \in \Gamma \}.
	\)
	From the proof of Lemma \ref{lem:R-next is serial}, it follows that $\Delta \in \MCS$ and $\overline{\Gamma} R_\lnext \overline{\Delta}$. On the other hand, if $\phi \in Y$, then $\lnext \phi \in \Gamma$, and hence $\phi \in \overline{\Delta}$. Thus $Y \subseteq \overline{\Delta}$, which immediately implies $Y = \overline{\Delta}$, by the $\chi$-maximality of $Y$. Therefore, $X R_\lnext Y$. 
	\qed
\end{proof}

\begin{lemma}\label{lem:equivalence definition of R-next by sprevious}
	Let $X, Y \in \MCS_\chi$.
	\[
	X R_\lnext Y 
	\quad\text{ if{f} }\quad 
	\nvdash_\CS \bigwedge Y \to \lnot \sprevious \bigwedge X.
	\]
\end{lemma}
\begin{proof}
	The proof follows from Lemma \ref{lem:equivalence definition of R-next} and the following fact
	\begin{align*}
		\vdash_\CS \phi \to \neg \lnext \psi \quad\mbox{iff}\quad \vdash_\CS \psi \to \neg \sprevious \phi. 
	\end{align*}
	\qed
\end{proof}

\begin{lemma}\label{lem:corollary of R-next}
	Let $X,Y \in \MCS_\chi$, $X R_\lnext Y$, and $\phi \in A_\chi$.
	
	\begin{enumerate}
		\item  If\/ $X \vdash_\CS \lnext \phi$, then $\phi \in Y$.
		
		\item  If\/ $X \vdash_\CS \neg \lnext \phi$, then $\neg \phi \in Y$.
	\end{enumerate}
\end{lemma}
\begin{proof}
	\begin{enumerate}
		\item 
		Suppose toward a contradiction that $\phi \not \in Y$. Thus, by Lemma \ref{lem: Facts about MCS-chi}, $\neg \phi \in Y$. Since $X \vdash_\CS \lnext \phi$, we have
		\mbox{$
			\vdash_\CS \bigwedge X \to \lnext \phi.
			$}
		Hence 
		$
		\vdash_\CS \bigwedge X \to \lnext \neg\neg \phi.
		$
		Therefore 
		$\vdash_\CS \bigwedge X \to \lnext \neg \bigwedge Y.
		$
		Thus 
		\[
		\vdash_\CS \bigwedge X \to \neg \lnext  \bigwedge Y,
		\]
		which would contradict $X R_\lnext Y$.
		
		\item Similar to the proof of part 1.    \qed
	\end{enumerate}
\end{proof}

\begin{lemma}\label{l:4.4}
	Let $X \in \MCS_\chi$ and let 
	$R_\lnext(X) \colonequals \{ Y \in  \MCS_\chi \ |\ X R_\lnext  Y\}$.
	We have 
	\[
	\vdash_\CS \bigwedge X \to \lnext \bigvee \big\{ \bigwedge Y \ |\ Y \in R_\lnext(X) \big\}.
	\]
\end{lemma}
\begin{proof} 
	By Lemma \ref{lem:equivalence definition of R-next}, for all $X, Y \in \MCS_\chi$ we have
	\begin{equation*}
		(\text{not } X R_\lnext Y) 
		\quad\text{implies}\quad
		\vdash_\CS \bigwedge X \to \lnot \lnext \bigwedge Y.
	\end{equation*}
	Thus 
	\[
	\vdash_\CS \bigwedge X \to \bigwedge\{ \lnot \lnext \bigwedge Y~|~Y \in \MCS_\chi \text{~and not } X R_\lnext Y\},
	\]
	and hence
	\begin{equation}\label{eq:next:1}
		\vdash_\CS \bigwedge X \to \neg \bigvee \{ \lnext \bigwedge Y~|~ Y \in \MCS_\chi \text{~and not } X R_\lnext Y\}.
	\end{equation}
	We also have (cf. \cite[Lemma 4.1]{HvdMV04})
	\begin{equation}\label{eq:next:2}
		\vdash_\CS \bigvee \big\{ \bigwedge Y \ |\ Y \in  \MCS_\chi \big\}.
	\end{equation}
	From \eqref{eq:next:2} by $\nextnecrule$ we get
	\[
	\vdash_\CS \lnext\bigvee \big\{ \bigwedge Y \ |\ Y \in  \MCS_\chi \big\}.
	\]
	By item 2 of Lemma \ref{lem: properties of previous operators}, we get
	\begin{equation*}
		\vdash_\CS \bigvee \big\{ \lnext \bigwedge Y \ |\ Y \in  \MCS_\chi \big\}.
	\end{equation*}
	Hence
	\[
	\vdash_\CS \bigwedge X \to \bigvee \big\{ \lnext \bigwedge Y \ |\ Y \in  \MCS_\chi \big\},
	\]
	from which it follows that
	\[
	\vdash_\CS \bigwedge X \to \bigvee \big\{ \lnext \bigwedge Y \ |\ Y \in  \MCS_\chi \text{~and~} X R_\lnext Y \big\} \vee  \bigvee \big\{ \lnext \bigwedge Y \ |\ Y \in  \MCS_\chi \text{~and not~} X R_\lnext Y \big\}.
	\]
	By \eqref{eq:next:1} we infer
	\[
	\vdash_\CS \bigwedge X \to \bigvee \big\{ \lnext \bigwedge Y \ |\ Y \in  \MCS_\chi  \text{ and } X R_\lnext Y \big\}
	\]
	and thus, by item 2 of Lemma \ref{lem: properties of previous operators}, we get
	\[
	\vdash_\CS \bigwedge X \to \lnext \bigvee \big\{\bigwedge Y \ |\ Y \in  R_\lnext(X)\}. 
	\]
	\qed
\end{proof}

\begin{lemma}\label{lem:wprevious-relation}
	Let $X \in \MCS_\chi$ and let 
	$R_\lnext^{-1}(X) \colonequals \{ Y \in  \MCS_\chi \ |\ Y  R_\lnext X \}$.
	We have 
	\[
	\vdash_\CS \bigwedge X \to \wprevious \bigvee \big\{ \bigwedge Y \ |\ Y \in R_\lnext^{-1}(X) \big\}.
	\]
\end{lemma}
\begin{proof} 
	First note that if $X$ is initial, i.e. $\wprevious \bot \in X$, then by Lemma \ref{lem:initial state} the set $R_\lnext^{-1}(X)$ is empty, and moreover we have 
	\[
	\vdash_\CS \bigwedge X \to \wprevious \bot,
	\]
	which implies that 
	\[
	\vdash_\CS \bigwedge X \to \wprevious \bigvee \big\{ \bigwedge Y \ |\ Y \in \emptyset \big\}.
	\]
	Now suppose	$X$ is not initial, and thus  by Lemma \ref{lem:initial state} there is $Y_0 \in \MCS_\chi$ such that $Y_0 R_\lnext X$. By Lemma \ref{lem:equivalence definition of R-next by sprevious}, for all $X, Y \in \MCS_\chi$ we have 
	\begin{equation*}
		(\text{not } Y R_\lnext X) 
		\quad\text{implies}\quad
		\vdash_\CS \bigwedge X \to \lnot \sprevious \bigwedge Y.
	\end{equation*}
	Thus 
	\[
	\vdash_\CS \bigwedge X \to \bigwedge\{ \lnot \sprevious \bigwedge Y~|~Y \in \MCS_\chi \text{~and not }  Y R_\lnext X\},
	\]
	and hence
	\begin{equation*}
		\vdash_\CS \bigwedge X \to \neg \bigvee \{ \sprevious \bigwedge Y~|~ Y \in \MCS_\chi \text{~and not } Y R_\lnext X\}.
	\end{equation*}
	Therefore
	\begin{equation}\label{eq:sprevious:1}
		\vdash_\CS \bigwedge X \to \neg \bigvee \{ \sprevious \bigwedge Y~|~ Y \in \MCS_\chi \text{~and not } Y R_\lnext X \text{ and } Y \neq Y_0 \}.
	\end{equation}
	From \eqref{eq:next:2}  by $\prevnecrule$  we get
	\[
	\vdash_\CS \wprevious \bigvee \big\{ \bigwedge Y \ |\ Y \in  \MCS_\chi \big\},
	\]
	and thus by item 3 of Lemma \ref{lem: properties of previous operators} we have
	\begin{equation}\label{eq:sprevious:2}
		\vdash_\CS \bigvee \big\{ \sprevious \bigwedge Y \ |\ Y \in  \MCS_\chi  \text{ and }  Y \neq Y_0 \big\} \vee \wprevious \bigwedge Y_0.
	\end{equation}
	By \eqref{eq:sprevious:1} and \eqref{eq:sprevious:2} we infer
	\[
	\vdash_\CS \bigwedge X \to \bigvee \big\{ \sprevious \bigwedge Y \ |\ Y \in  \MCS_\chi  \text{ with } Y  R_\lnext X  \text{ and } Y \neq Y_0 \big\} \vee \wprevious \bigwedge Y_0
	\]
	and thus by item 2 of Lemma \ref{lem: properties of previous operators}
	\[
	\vdash_\CS \bigwedge X \to \sprevious \bigvee \big\{\bigwedge Y \ |\ Y \in  \MCS_\chi  \text{ with } Y  R_\lnext X   \text{ and } Y \neq Y_0 \big\} \vee \wprevious \bigwedge Y_0.
	\]
	Thus, by the axiom $\swprevax$, we get
	\[
	\vdash_\CS \bigwedge X \to \wprevious \bigvee \big\{\bigwedge Y \ |\ Y \in  \MCS_\chi  \text{ with } Y  R_\lnext X   \text{ and } Y \neq Y_0 \big\} \vee \wprevious \bigwedge Y_0.
	\]
	Therefore, by item 2 of Lemma \ref{lem: properties of previous operators}, we get
	\[
	\vdash_\CS \bigwedge X \to \wprevious \bigvee \big\{\bigwedge Y \ |\ Y \in  \MCS_\chi  \text{ with }  Y R_\lnext X \big\},
	\]
	and thus
	\[
	\vdash_\CS \bigwedge X \to \wprevious \bigvee \big\{ \bigwedge Y \ |\ Y \in R_\lnext^{-1}(X) \big\}.  
	\]
	\qed
\end{proof}

\begin{definition}
	A finite sequence $(X_0, X_1, \ldots, X_n)$ of elements of~$\MCS_\chi$ is called a \emph{$\phi \luntil \psi$-sequence starting with $X$} if
	\begin{enumerate}
		\item $X_0 = X$,
		
		\item $X_j R_\lnext X_{j+1}$, for all $0 \leq j < n$, 
		
		\item $\psi \in X_n$,
		
		\item $\phi \in X_j$, for all $0 \leq j < n$.
	\end{enumerate}
\end{definition}

\begin{lemma}\label{lem:until-sequence}
	For every $X \in \MCS_\chi$, if $\phi \luntil \psi \in X$, then there exists a  $\phi \luntil \psi$-sequence starting with $X$.
\end{lemma}
\begin{proof}
	Let $X = \overline{\Gamma}$, for some $\Gamma \in \MCS$. Suppose $\phi \luntil \psi \in X$ and there exists no  $\phi \luntil \psi$-sequence starting with $X$. Since $\phi \luntil \psi \in  \Subf^+(\chi)$, we have $\phi \luntil \psi \in  A_\chi$, and thus $\phi, \psi \in  A_\chi$. We first show that:
	\begin{equation}\label{eq:properties of Gamma in T-Until}
		\neg \psi \in X \text{ and } \phi \in X.
	\end{equation}
	Suppose  $\psi \in X$. Then the sequence $(X)$ would be a $\phi \luntil \psi$-sequence starting with $X$, contradicting our assumption. Thus $\psi \not \in X$, and hence by  Lemma \ref{lem: Facts about MCS-chi} we get $\neg \psi \in X$.  On the other hand, from $\phi \luntil \psi \in \overline{\Gamma}$ it follows that $\psi \vee (\phi \wedge \lnext (\phi \luntil \psi) ) \in \Gamma$. From this we can immediately deduce by $\neg \psi \in \Gamma$ that $\phi \in \Gamma$, and hence $\phi \in X$.
	
	Let $T_\luntil$ be the smallest set of elements of $\MCS_\chi$ such that
	\begin{enumerate}
		\item $X \in T_\luntil$;
		\item for each $Z \in \MCS_\chi$, if there is $Y \in T_\luntil$ such that $Y R_\lnext Z$ and $\phi \in Z$, then $Z \in T_\luntil$.
	\end{enumerate}
	First we show that for all $Z \in T_\luntil$ such that $Z \neq X$ and $\phi \in Z$ there is $Y \in T_\luntil$ such that $Y R_\lnext Z$. In order to prove this let 
	\[
	T = \{ Z \in T_\luntil ~|~ Z \neq X, \phi \in Z, \text{ and there is no } Y \in T_\luntil \text{ such that }  Y R_\lnext Z \}.
	\]
	If $T \neq \emptyset$, then $T_\luntil \setminus T$ is a proper subset of $T_\luntil$ that satisfies properties 1 and 2. This contradicts the fact that $T_\luntil$ is the smallest set with properties 1 and 2. Thus $T = \emptyset$.
	
	From the definition of $T_\luntil$ and \eqref{eq:properties of Gamma in T-Until}, it is not difficult to show that $\phi \in Z$ for all $Z \in T_\luntil$. Thus it follows that for all $Z \in T_\luntil$ there exists $X_0, \ldots,X_n \in T_\luntil$, for $n \geq 0$, such that $X_0 = X$, $X_n = Z$,  $\phi \in X_0 \cap \ldots \cap X_n$ and $X_0 R_\lnext \ldots R_\lnext X_n.$
	
	Now we claim that $\neg \psi \in Z$, for all $Z \in T_\luntil$. 
	First note that, by \eqref{eq:properties of Gamma in T-Until}, $\neg \psi \in X$. For $Z \in T_\luntil$ such that $Z \neq X$, there exists $X_0, \ldots,X_n \in T_\luntil$, for $n > 0$, such that $X_0 = X$, $X_n = Z$, $\phi \in X_0 \cap \ldots \cap X_n$, and 
	$X_0 R_\lnext \ldots R_\lnext X_n$. Thus $\psi \not \in Z$, since otherwise $
	(X_0, \ldots,X_n)
	$ would be a $\phi \luntil \psi$-sequence starting with $X$, contradicting our assumption. Therefore, by Lemma \ref{lem: Facts about MCS-chi}, $\neg \psi \in Z$. This completes the proof of the claim.
	
	Let 
	\[
	\rho \colonequals \bigvee \big\{ \bigwedge Y \ |\ Y \in T_\luntil \big\}.
	\]
	Using the above claim we get $\vdash_\CS \rho \to \lnot \psi$. 
	
	Let $Y \in T_\luntil$ and  $Z \in \MCS_\chi$ such that $Y R_\lnext Z$. We have either $\phi \in Z$ or $\phi \not \in Z$. 
	If $\phi \in Z$, then by property 2 we have $Z \in T_\luntil$, and hence $\vdash_\CS \bigwedge Z \to \rho$. 
	If $\phi \not \in Z$, then $\neg \phi \in Z$. In addition, $\psi \not \in Z$, since otherwise we get a $\phi \luntil \psi$-sequence starting with $X$. Thus $\neg \psi \in Z$, and hence $\vdash_\CS \bigwedge Z \to \neg \phi \wedge \neg \psi$.
	
	Thus, for each $Y \in T_\luntil$ and each $Z \in \MCS_\chi$ such that $Y R_\lnext Z$, we have
	\[
	\text{either}\quad
	\vdash_\CS \bigwedge Z \to \rho \quad \text{or} \quad \vdash_\CS \bigwedge Z \to \lnot \phi \land \lnot \psi,
	\]
	and hence, 
	\begin{equation}\label{eq: until-seq}
		\vdash_\CS \bigwedge Z \to \rho \vee (\lnot \phi \land \lnot \psi).
	\end{equation}

	By Lemma~\ref{l:4.4}, for each $Y \in T_\luntil$ we have 
	\[
	\vdash_\CS \bigwedge Y \to \lnext (\bigwedge Z_1  \vee \ldots \vee \bigwedge Z_n)
	\]
	such that $Z_i \in \MCS_\chi$ and $Y R_\lnext Z_i$, for $i=1,\ldots,n$. By \eqref{eq: until-seq}, we get 
	\[
	\vdash_\CS \bigwedge Y \to \lnext (\rho \vee (\neg \phi \wedge \neg \psi)),
	\]
	for each $Y \in T_\luntil$. Thus
	$
	\vdash_\CS \rho \to \lnext (\rho \lor (\lnot \phi \land \lnot \psi)).
	$
	Using $\uindrule$, we obtain $\vdash_\CS \rho \to \lnot(\phi \luntil \psi)$.
	Since $X \in T_\luntil$, this implies  
	$
	\vdash_\CS \bigwedge X \to \lnot(\phi \luntil \psi),
	$
	which contradicts the assumption
	$\phi \luntil \psi \in X$.  \qed
\end{proof}

\begin{definition}
	A finite sequence $(X_0, X_1, \ldots, X_n)$ of elements of~$\MCS_\chi$ is called a \emph{$\phi \lsince \psi$-sequence ending with $X$} if
	\begin{enumerate}
		\item $X_n = X$,
		
		\item $X_j R_\lnext X_{j+1}$, for all $0 \leq j < n$, 
		
		\item $\psi \in X_0$,
		
		\item $\phi \in X_j$, for all $0< j \leq n$.
	\end{enumerate}
\end{definition}

\begin{lemma}\label{lem:since-sequence}
	For every $X \in \MCS_\chi$, if $\phi \lsince \psi \in X$, then there exists a  $\phi \lsince \psi$-sequence ending with $X$.
\end{lemma}
\begin{proof}
	Let $X = \overline{\Gamma}$, for some $\Gamma \in \MCS$. Suppose $\phi \lsince \psi \in X$ and there exists no  $\phi \lsince \psi$-sequence ending with $X$. Since $\phi \lsince \psi \in  \Subf^+(\chi)$, we have  $\phi,\psi \in  A_\chi$. From this, similar to the proof of Lemma \ref{lem:until-sequence}, it is proved that $\phi, \neg \psi \in X$.
	
	Let $T_\lsince$ be the smallest set of elements of $\MCS_\chi$ such that
	\begin{enumerate}
		\item $X \in T_\lsince$;
		\item for each $Y \in \MCS_\chi$, if there is $Z \in T_\lsince$ such that $Y R_\lnext Z$ and $\phi \in Y$, then $Y \in T_\lsince$.
		
	\end{enumerate}
	
	Similar to the proof of Lemma \ref{lem:until-sequence}, it is proved that
	for all $Y \in T_\lsince$ there exists $n \geq 0$ and there exists $X_0, \ldots,X_n \in T_\lsince$ such that $X_0 = Y$, $X_n = X$,  and $X_0 R_\lnext \ldots R_\lnext X_n.$ 	Moreover, it is not difficult to show that $\neg \psi, \phi \in Y$, for all $Y \in T_\lsince$.

	Let 
	\[
	\rho \colonequals \bigvee \big\{ \bigwedge Z \ |\ Z \in T_\lsince \big\}.
	\]
	We have $\vdash_\CS \rho \to \lnot \psi$. In addition, for each $Z \in T_\lsince$ and each $Y \in \MCS_\chi$ with $Y R_\lnext Z$, we have 
	\[
	\text{either}\quad
	\vdash_\CS \bigwedge Y \to \rho \quad \text{or} \quad \vdash_\CS \bigwedge Y \to \lnot \phi \land \lnot \psi,
	\]
	and hence, 
	\begin{equation}\label{eq:since-seq}
		\vdash_\CS \bigwedge Y \to \rho \vee (\lnot \phi \land \lnot \psi).
	\end{equation}

	By Lemma \ref{lem:wprevious-relation}, for each $Z \in T_\lsince$ we have 
	\[
	\vdash_\CS \bigwedge Z \to \wprevious (\bigwedge Y_1 \vee \ldots \vee \bigwedge Y_n)
	\]
	such that $Y_i \in \MCS_\chi$ and $Y_i R_\lnext Z$, for $i=1,\ldots,n$. By \eqref{eq:since-seq}, we get 
	\[
	\vdash_\CS \bigwedge Z \to \wprevious (\rho \vee (\neg \phi \wedge \neg \psi)),
	\]
	for each $Z \in T_\lsince$. Thus
	$
	\vdash_\CS \rho \to \wprevious (\rho \lor (\lnot \phi \land \lnot \psi)).
	$
	Using $\sindrule$, we obtain $\vdash_\CS \rho \to \lnot(\phi \lsince \psi)$.
	Since $X \in T_\lsince$, this implies  
	$
	\vdash_\CS \bigwedge X \to \lnot(\phi \lsince \psi),
	$
	which contradicts the assumption
	$\phi \lsince \psi \in X$.  \qed
\end{proof}

\begin{definition}\label{Def:acceptable sequence}
	
	An infinite sequence $(X_0, X_1, \ldots)$ of elements of $\MCS_\chi$ is called \textit{acceptable} (for $\mathsf{L}_\CS$) if 
	\begin{enumerate}
		\item $X_n R_\lnext X_{n+1}$ for all $n \geq 0$, and
		\item for all $n$, if $\phi \luntil \psi \in X_n$, then there exists $m \geq n$ such that $\psi \in X_m$ and $\phi \in X_k$ for all $k$ with $n \leq k <m$.
		\item $\wprevious \bot \in X_0$ (i.e. $X_0$ is initial). 
	\end{enumerate}
\end{definition}

\begin{lemma}\label{lem:finite seq to acceptable seq}
	Every finite sequence $(X_0, X_1, \ldots, X_n)$ of elements of\/ $\MCS_\chi$ with $\wprevious \bot \in X_0$ and $X_j R_\lnext X_{j+1}$, for all $0 \leq j < n$, can be extended to an  acceptable sequence.  
\end{lemma}
\begin{proof}
	In order to fulfill the requirements of Definition \ref{Def:acceptable sequence}, we shall extend the sequence $(X_0, X_1, \ldots, X_n)$ by the following steps.
	
	Suppose $\phi \luntil \psi \in X_0$. Then either $\psi \in X_0$ or $\neg\psi \in X_0$. In the former case the requirement is fulfilled for the formula $\phi \luntil \psi$  in $X_0$, and we go to the next step. In the latter case, using axiom $(\luntil 2)$, $X_0 \vdash_\CS \phi \wedge \lnext (\phi \luntil \psi)$. 
	Since $X_0 R_\lnext X_1$, By Lemma \ref{lem:corollary of R-next}, we get $\phi \luntil \psi \in X_1$.

	We can repeat this argument for $X_i$ for 
	$1\leq i \leq n$. We find that the requirement for $\phi \luntil \psi \in X_0$ is either fulfilled in $(X_0, X_1, \ldots, X_n)$ or we get $\phi \luntil \psi \in X_n$ and  $\phi \in X_i$ for $1\leq i \leq n$. 
	In the latter case, by Lemma~\ref{lem:until-sequence}, there exists a sequence $(X_n, X_{n+1}, \ldots, X_{n+m})$
	such that 
	$\phi \in X_i$ for $n\leq i < n+m$,   
	$\psi \in X_{n+m}$, 
	and  $X_i R_\lnext X_{i+1}$ for $n\leq i < n+m$. This gives a finite extension of the original sequence that satisfies the requirement imposed by $\phi \luntil \psi \in X_0$.
	
	In the next step we repeat this argument for the remaining $\luntil$-formulas at $X_0$. Eventually we obtain a finite sequence that satisfies all requirements imposed by $\luntil$-formulas at $X_0$.
	
	We may move on to $X_1$ and apply the same procedure. 
	It is clear that by iterating the above argument to all $\luntil$-formulas of all elements $X_i$ of the sequence $(X_0, X_1, \ldots, X_n,\ldots)$, including $\luntil$-formulas of new elements $X_i$ for $i > n$, we obtain in the limit a (finite or infinite) sequence that extends $(X_0, X_1, \ldots, X_n)$ and satisfies conditions 1--3 of Definition \ref{Def:acceptable sequence}. If the resulting sequence is finite, then by seriality of  $R_\lnext$ it can be extended to an infinite sequence, and in each step of this extension we can repeat the above argument to fulfill the obligations arising from the $\luntil$-formulas.  	 
	Thus, we finally get an acceptable sequence that extends $(X_0, X_1, \ldots, X_n)$. \qed
\end{proof}

\begin{corollary}\label{cor:acceptable sequence starts with Gamma}
	For every $X \in \MCS_\chi$, there is an acceptable sequence containing $X$. 
\end{corollary}
\begin{proof}
	Given $X \in \MCS_\chi$, since $\lonce \wprevious \bot = \top \lsince \wprevious \bot \in X$, by Lemma \ref{lem:since-sequence}, there exists a $\top \lsince \wprevious \bot$-sequence  $(X_0, X_1, \ldots, X_n)$ ending with $X$, i.e. $X_n = X$, $X_j R_\lnext X_{j+1}$, for all $0 \leq j < n$, and $\wprevious \bot \in X_0$. By Lemma \ref{lem:finite seq to acceptable seq}, this sequence can be extended to an acceptable sequence containing $X$.  \qed 
\end{proof}

Let $(X_0, X_1, \ldots)$ be an acceptable sequence of elements of $\MCS_\chi$. Then there exist $\Gamma_{0}, \Gamma_{1}, \ldots \in \MCS$ and $\Delta_{1}, \Delta_{2}, \ldots \in \MCS$ such that
$$\RO{X_{0}}{X_{1}}{\overline{\Gamma_{0}}}{\overline{\Delta_{1}}}, \RO{X_{1}}{X_{2}}{\overline{\Gamma_{1}}}{\overline{\Delta_{2}}},\RO{X_{2}}{X_{3}}{\overline{\Gamma_{2}}}{\overline{\Delta_{3}}}, \ldots.$$

Note that for each $i \geq 1$ we have $\RO{X_{i-1}}{X_{i}}{\overline{\Gamma_{i-1}}}{\overline{\Delta_{i}}}$, and thus $X_0 = \overline{\Gamma_{0}}$ and $X_j = \overline{\Gamma_{j}} = \overline{\Delta_{j}}$ for all $j > 0$.

\begin{lemma}\label{lem: MCS temporal properties}
	Let $(X_0, X_1, \ldots)$ be an acceptable sequence of elements of $\MCS_\chi$, let $n \geq 0$, and let $\RO{X_{i-1}}{X_{i}}{\overline{\Gamma_{i-1}}}{\overline{\Delta_{i}}}$, for  $i \geq 1$.
	\begin{enumerate}
		\setlength\itemsep{0.01cm}
		\item 
		\begin{enumerate}
			\item If $\phi \lsince \psi \in X_n$, then there exists $m \leq n$ such that $\psi \in X_m$ and $\phi \in X_k$ for all $k$ with $m < k \leq n$.
			
			\item If $\phi \lsince \psi \in \Subf^+(\chi)$ and there exists $m \leq n$ such that $\psi \in X_m$ and $\phi \in X_k$ for all $k$ with $m < k \leq n$, then $\phi \lsince \psi \in X_n$.
		\end{enumerate}
		
		\item If $\leventually \phi \in  X_n$, then $\phi \in  X_m$ for some $m \geq n$. 
		
		\item 
		\begin{enumerate}
			\item If $\lalways \phi \in X_n$, then $\phi \in X_m$ for all $m \geq n$. 
			
			\item If $\lalways \phi \in \Subf^+(\chi)$ and $\phi \in X_m$ for all $m \geq n$, then $\lalways \phi \in X_n$.
		\end{enumerate}
		
		\item If $\lonce \phi \in X_n$, then $\phi \in X_{m}$ for some $m \leq n$. 
		
		\item 
		\begin{enumerate}
			\item If $\lsofar \phi \in X_n$, then $\phi \in X_m$ for all $m \leq n$.
			
			\item If $\lsofar \phi \in \Subf^+(\chi)$ and $\phi \in X_m$ for all $m \leq n$, then $\lsofar \phi \in X_n$.
		\end{enumerate}

		
		\item If $n > 0$ and $\wprevious \phi \in X_n$, then $\phi \in X_{n-1}$.
		
		\item If $n > 0$ and $\sprevious \phi \in X_n$, then $\phi \in X_{n-1}$.
	\end{enumerate}
	
\end{lemma}
\begin{proof}
	
	\begin{enumerate}
		
		\item The proof involves a routine induction on $n$. We prove item (a). The proof of (b) is similar.
		
		Suppose $n=0$ and $\phi \lsince \psi \in X_0$. Since $\phi \lsince \psi \in \Gamma_0$, using axiom $(\lsince 2)$ we have either $\psi \in \Gamma_0$ or $\phi \wedge \sprevious (\phi \lsince \psi) \in \Gamma_0$. In the former case, we get $\psi \in X_0$ and we are done. By Lemma \ref{lem:initial state does not contain strong previous formulas} the latter case cannot happen. 
		
		Suppose $n>0$ and $\phi \lsince \psi \in X_n$. Since $\phi \lsince \psi \in \Delta_n$, then using axiom $(\lsince 2)$ we have either $\psi \in \Delta_n$ or $\phi \wedge \sprevious (\phi \lsince \psi) \in \Delta_n$. In the former case, we get $\psi \in X_n$ and we are done. In the latter case, we have $\phi \in \Delta_n$, and thus $\phi \in X_n$. On the other hand, since $\RO{X_{n-1}}{X_{n}}{\overline{\Gamma_{n-1}}}{\overline{\Delta_{n}}}$, by Lemma \ref{lem:R-next} we get $\phi \lsince \psi \in \Gamma_{n-1}$, and hence $\phi \lsince \psi \in X_{n-1}$. By the induction hypothesis there exists $m \leq n-1$ such that $\psi \in X_m$ and $\phi \in X_k$ for all $k$ with $m < k \leq n-1$. This completes the proof of the only if direction. 
		
		
		\item Suppose $\leventually \phi \in  X_n$. Then $ \top \luntil  \phi \in X_n$. Since $(X_0, X_1, \ldots)$ is an acceptable sequence, there exists $m \geq n$ such that $\phi \in X_m$.
		
		\item For item (a), suppose $\lalways \phi \in X_n$. By induction on $m$ we show that  $\lalways \phi \in X_m$, for all $m \geq n$. Having proved this, by Lemma \ref{lem:mix:1}, we get $\phi \in X_m$, for all $m \geq n$, as desired. 
		
		The base case $m = n$ is trivial. Suppose that $\lalways \phi \in X_m$, for some $m > n$. We want to show that $\lalways \phi \in X_{m+1}$. Since $\lalways \phi \in \Gamma_m$, by Lemma \ref{lem:mix:1} we get $\lnext \lalways \phi \in \Gamma_m$. Since $\RO{X_{m}}{X_{m+1}}{\overline{\Gamma_{m}}}{\overline{\Delta_{m+1}}}$, by Lemma \ref{lem:R-next}, we obtain $\lalways \phi \in \Delta_{m+1}$.  Hence $\lalways \phi \in X_{m+1}$. 
		
		For item (b), suppose that $\lalways \phi \in \Subf^+(\chi)$, $\phi \in X_m$ for all $m \geq n$, and $\lalways \phi \not \in X_n$. Thus $\neg \lalways \phi \in X_{n}$, and hence $\leventually \neg \phi \in X_{n}$. By item 2 above, $\neg \phi \in X_m$ for some $m \geq n$, which would contradict the assumption.

		\item Suppose $\lonce \phi \in X_n$. Then $\top \lsince \phi \in X_{n}$. By item 1, we get $\phi \in X_{m}$ for some $m \leq n$. 
		
		\item For item (a), suppose $\lsofar \phi \in X_n$. By induction on $m$ we show that  $\lsofar \phi \in X_m$, for all $m \leq n$. Having proved this, by Lemma \ref{lem:mix:1}, we get $\phi \in X_m$, for all $m \leq n$, as desired. 
		
		The base case $m = n$ is trivial. Suppose that $\lsofar \phi \in X_m$, for some $m < n$. We want to show that $\lsofar \phi \in X_{m-1}$. Since $\lsofar \phi \in \Delta_m$, by Lemma \ref{lem:mix:1} we get $\wprevious \lsofar \phi \in \Delta_m$. Since $\RO{X_{m-1}}{X_{m}}{\overline{\Gamma_{m-1}}}{\overline{\Delta_{m}}}$, by Lemma \ref{lem:R-next}, we obtain $\lsofar \phi \in \Gamma_{m-1}$.  Hence $\lsofar \phi \in X_{m-1}$. 
		
		For item (b), suppose that $\lsofar \phi \in \Subf^+(\chi)$, $\phi \in X_m$ for all $m \leq n$, and $\lsofar \phi \not \in X_n$. Thus $\neg \lsofar \phi \in X_{n}$, and hence $\lonce \neg \phi \in X_{n}$. By item 4 above, $\neg \phi \in X_m$ for some $m \leq n$, which would contradict the assumption.
		
		\item Suppose $n > 0$ and $\wprevious \phi \in X_n$. Since $\RO{X_{n-1}}{X_{n}}{\overline{\Gamma_{n-1}}}{\overline{\Delta_{n}}}$ and $\wprevious \phi \in \Delta_n$, by Lemma \ref{lem:R-next},  it follows that $\phi \in \Gamma_{n-1}$, and hence $\phi \in X_{n-1}$. 
		
		\item The proof is similar to the proof of item 6.\qed
		
	\end{enumerate}
\end{proof}

\section{Semantics of $\LPLTLp$}
\label{sec:Semantics}

In this section we introduce interpreted systems based on  Fitting-models as semantics for temporal justification logic $\LPLTLp$.

\begin{definition}\label{def:frame-run-system}
	A \emph{frame} is a tuple $(S, R_1,\ldots,R_\numberofagents)$ where
	\begin{enumerate}
		\item $S$ is a non-empty set of states;
		\item each $R_i \subseteq S \times S$ is a reflexive and transitive relation. 
	\end{enumerate}
	A \emph{run}~$r$ on a frame is a function from $\N$ to states, i.e., $r: \N \to S$. A \emph{system}~$\runs$ is a non-empty set of runs. 
	
	Given a run $r$ and $n \in \N$, the pair $(r,n)$ is called a \emph{point}.
\end{definition}

\begin{definition}\label{def:evidence function for LPLTL}
	Given a frame $(S, R_1,\ldots,R_\numberofagents)$,
	a \emph{$\CS$-evidence function for agent~$\agent$} is a function 
	\[
	\evidence_i: S \times \Terms \to \powerset(\Formulae)
	\]
	satisfying the following conditions.
	For all terms $s,t \in \Terms$, all formulas $\phi,\psi \in \Formulae$, all $v,w \in S$, and all $i \in \Ag$:
	\begin{enumerate}
		\item 
		$\evidence_\agent(v,t) \subseteq \evidence_\agent(w, t)$, whenever $R_i(v,w)$; \hfill (monotonicity)
		\item 
		if $\jbox{t}_\agent \phi \in \CS$, then $\phi \in \evidence_\agent(w,t)$; \hfill (constant specification)

		\item 
		if $\phi \limplies \psi \in \evidence_\agent(w,t)$ and $\phi \in \evidence_\agent(w,s)$, then $\psi \in \evidence_\agent(w, t \tapp s)$; \hfill (application)
		
		\item 
		$\evidence_\agent(w,s) \cup \evidence_\agent(w,t) \subseteq \evidence_\agent(w,s + t)$; \hfill (sum)
		
		\item 
		if $\phi \in \evidence_\agent(w,t)$, then $\jbox{t}_\agent \phi \in \evidence_\agent(w,\tinspect t)$. \hfill (positive introspection)
	\end{enumerate}
\end{definition}

\begin{definition}\label{def:interpreted sysytems LPLTLp}
	An  \emph{interpreted system for $\LPLTLp_\CS$} (or for $\CS$) is a tuple 
	\[
	\system = (\runs, S, R_1,\ldots,R_\numberofagents, \evidence_1\ldots,\evidence_\numberofagents, \valuation)
	\] 
	where
	\begin{enumerate}
		\item $(S, R_1,\ldots,R_\numberofagents)$ is a frame;
		\item $\runs$ is a system on that frame;
		\item $\evidence_\agent$ is a $\CS$-evidence function for agent~$\agent$ for $1 \leq \agent \leq \numberofagents$;
		\item $\valuation: S \to \powerset(\Prop)$ is a valuation.
	\end{enumerate}
\end{definition}

\begin{definition}\label{def:truth conditions interpreted systems}
	Given an interpreted system 
	\[
	\system = (\runs, S, R_1,\ldots,R_\numberofagents, \evidence_1,\ldots,\evidence_\numberofagents, \valuation),
	\] 
	a run $r \in \runs$, and  $n \in \N$, we define truth of a formula $\phi$ in $\system$ at point $(r,n)$ inductively by 
	\begin{align*}
	(\system, r, n) &\entails P \text{ iff } P \in \valuation(r(n)) \, ,\\
	(\system, r, n) &\not\entails \lfalse \, ,\\
	(\system, r, n) &\entails \phi \limplies \psi \text{ iff } (\system, r, n) \not\entails \phi \text{ or } (\system, r, n) \entails \psi \, ,\\
	(\system, r, n) &\entails \wprevious \phi \text{ iff $n=0$ or } (\system, r, n-1) \entails \phi \, ,\\
	(\system, r, n) &\entails \lnext \phi \text{ iff } (\system, r, n+1) \entails \phi \, ,\\
	(\system, r, n) &\entails \phi \lsince \psi \text{ iff there is some } m \leq n \text{ such that } (\system, r, m) \entails \psi \\ & \qquad\qquad \text{ and } (\system, r, k) \entails \phi \text{ for all $k$ with } m < k \leq n \, ,\\
	(\system, r, n) &\entails \phi \luntil \psi \text{ iff there is some } m \geq n \text{ such that } (\system, r, m) \entails \psi \\ & \qquad\qquad \text{ and } (\system, r, k) \entails \phi \text{ for all $k$ with  } n \leq k < m \, ,\\    
	(\system, r, n) &\entails \jbox{t}_\agent \phi \text{ iff }  \phi \in \evidence_\agent(r(n),t)  \text { and } (\system, r^\prime, n^\prime) \entails \phi \\ &\qquad\qquad \text{ for all } r^\prime \in \runs \text{ and } n^\prime \in \N \text{ such that } R_\agent(r(n) , r^\prime(n^\prime)) \, .
	\end{align*}

	As usual, we write $\system \entails \phi$ if
	for all $r \in \runs$ and all $ n \in \N$, we have 
	$(\system, r, n) \entails \phi$.
	Further, we write $\entails_\CS \phi$ if $\system \entails \phi$ for all 
	interpreted systems $\system$ for $\CS$.
\end{definition}

\begin{definition}
 Given a set of formulas $\Gamma$ and a formula $\phi$, the (local) consequence relation is defined as follows: $\Gamma \models_\CS \phi$ iff for all 
 interpreted systems $\system = (\runs, \ldots)$ for $\CS$, for all $r \in \runs$, and for all $ n \in \N$, if $(\system, r, n) \entails \psi$ for all $\psi \in \Gamma$, then $(\system, r, n) \entails \phi$.
\end{definition}

From the above definitions it follows that:  
	\begin{align*}
	(\system, r, n) &\entails \leventually \phi \text{ iff } (\system, r, m) \entails \phi \text{ for some } m\geq n \, ,\\
	(\system, r, n) &\entails \lalways \phi \text{ iff } (\system, r, m) \entails \phi \text{ for all } m\geq n \, ,\\
	(\system, r, n) &\entails \lonce \phi \text{ iff } (\system, r, m) \entails \phi \text{ for some } m \leq n \, ,\\
	(\system, r, n) &\entails \lsofar \phi \text{ iff } (\system, r, m) \entails \phi \text{ for all } m \leq n \, ,\\
	(\system, r, n) &\entails \sprevious \phi \text{ iff $n>0$ and } (\system, r, n-1) \entails \phi \, .
	\end{align*}
 
It is sometime convenient to use the following truth conditions for since and until formulas, which are clearly equivalent to the corresponding conditions given in Definition \ref{def:truth conditions interpreted systems}.
 
 \begin{align*}
 (\system, r, n) &\entails \phi \lsince \psi \text{ iff there is some $m$ with } n \geq m \geq 0 \text{ such that } (\system, r, n-m) \entails \psi \\ & \qquad\qquad \text{ and } (\system, r, n-k) \entails \phi \text{ for all $k$ with } 0 \leq k < m \, \\
(\system, r, n) &\entails \phi \luntil \psi \text{ iff there is some } m \geq 0 \text{ such that } (\system, r, n+m) \entails \psi \\ & \qquad\qquad \text{ and } (\system, r, n+k) \entails \phi \text{ for all $k$ with } 0 \leq k < m \, .
\end{align*}

\begin{remark}\label{rem:time is n-interpreted systems}
	Note that 
	\[
	(\system, r, n) \entails \wprevious \bot \text{ iff $n=0$. }
	\]
	Thus $\wprevious \bot$ expresses the property ``the time is 0." Similarly, $\sprevious^m \wprevious\bot$, where $\sprevious^m$ is the iteration of $\sprevious$, $m$ times, expresses the property ``the time is m." 
	
	Let ``$\ltime=m$" abbreviate $\sprevious^m \wprevious \bot$ and $\ltrue_m (\phi)$ abbreviate $\ltime = m \to \phi$. It is easy to show that
	\[
	(\system, r, n) \entails \ltime=m  \text{ iff $n=m$ }
	\]
	and
	\[
	\system \models \ltrue_m (\phi) \text{ iff } \system, r, m \models \phi \quad \text{for all $r \in \mathcal{R}$}.
	\]
	Thus, $\ltrue_m (\phi)$ expresses that ``$\phi$ is true at time $m$."
\end{remark}

\begin{remark}
	The interpreted systems are originally formulated by means of the notions of \textit{local} and \textit{global} states (see e.g. \cite{FHMV95,HvdMV04}). Now I aim to define the interpreted systems for $\LPLTLp$, using the notions of local and global states, so that it more closely matches the original definition of interpreted systems given in \cite{FHMV95}.
	
	Suppose that at any point in time the system is in some global state, defined by the local states of the agents and the state of other objects of interest (which is refered to as the ``environment"). Let $\localstates$ be some set of \emph{local states}. Informally, an agent's local state captures all the information available to her at a given moment of time. A \emph{global state} is a $(\numberofagents+1)$-tuple $\langle l_e, l_1, \dots, l_\numberofagents \rangle \in \localstates^{\numberofagents+1}$, where $l_e$ is the state of environment and $l_i$ is the local state of agent $i$ for $i = 1, \ldots,h$. Now in order to  define the interpreted systems for $\LPLTLp$ using the notions of local and global states, 	it is enough to put the set of states $S := \localstates^{\numberofagents+1}$. 
	As before a \emph{run}~$r$ is a function from time to global states, i.e.,  $r: \N \to \localstates^{\numberofagents+1}$, and a \emph{system} is a set $\runs$ of runs. The  definitions of $\CS$-evidence functions, interpreted systems, and truth are as before. Note that here $\langle l_e, l_1, \dots, l_\numberofagents \rangle R_i \langle l'_e, l'_1, \dots, l'_\numberofagents \rangle$ means ``the local state $l'_i$ is epistemically possible for agent $i$ in the local state $l_i$.''
	
	It is  worth noting that the semantics given by Bucheli in \cite{Bucheli15} for temporal justification logic employs global states. However, there is a minor difference between Bucheli's semantics and ours. Since he modeled the knowledge part of the temporal justification logic by a justification counterpart of the modal logic {\sf S5}, he defines \textit{indistinguishability} relations $\accrel_\agent$ between points, for each agent $i$, which are clearly equivalence relations. In contrast to his formulation, our temporal justification logic is based on a justification counterpart of the modal logic {\sf S4}, and thus we naturally make use of reflexive and transitive accessibility relations $R_\agent$ for each agent $i$.
\end{remark}
\section{Soundness and Completeness of $\LPLTLp$}
\label{sec:Completeness}

The soundness proof for $\LPLTLp_\CS$ is a straightforward combination of the soundness proofs for temporal logic and justification logic by induction on the derivation.

\begin{theorem}[Soundness]
		For each formula $\phi$ and finite set of formulas $\Gamma$,
	\[
	 \Gamma\vdash_\CS \phi  \quad\text{implies}\quad  \Gamma \models_\CS  \phi.
	\]
\end{theorem}

For the completeness proof we employ the canonical model construction.

\begin{definition}\label{def:canonical interpreted systems for LPLTL}
	The $\chi$-canonical interpreted system 
	\[
	\system = (\runs, S, R_1,\ldots,R_\numberofagents, \evidence_1\ldots,\evidence_\numberofagents, \valuation)
	\] 
	for $\LPLTLp_\CS$ is defined as follows:
	\begin{enumerate}
		\item $\runs$ consists of all mappings $r: \N \to \MCS_\chi$ such that
		$(r(0), r(1), \ldots)$ is an acceptable sequence;
		\item $S \colonequals \MCS_\chi$;
		%
		
		\item $X R_\agent Y$ if{f} for all $\Delta \in \MCS$ such that $Y = \overline{\Delta}$ there exists $\Gamma \in \MCS$ such that $X = \overline{\Gamma}$ and  
		$\{ \phi \mid  \jbox{t}_\agent \phi \in \Gamma   \text{ for some $t$}\} \subseteq \Delta$; 

		\item $\evidence_\agent (X, t) \colonequals \{ \phi \mid \jbox{t}_\agent \phi \in \bigcap \{ \Gamma \in \MCS \mid X = \overline{\Gamma} \}$;
		\item $\valuation(X) \colonequals \Prop\cap X$.
	\end{enumerate}
\end{definition}

Note that $S = \MCS_\chi = \{ r(n) \ |\ r\in \runs, n \in \N \}$. 

\begin{lemma}\label{lem: canonical interpreted model is a model}
	The $\chi$-canonical interpreted system 
	\[
	\system = (\runs, S, R_1,\ldots,R_\numberofagents, \evidence_1\ldots,\evidence_\numberofagents, \valuation)
	\] 
	for $\LPLTLp_\CS$ is an interpreted system for $\LPLTLp_\CS$.
\end{lemma}
\begin{proof}
	It is not difficult to show that each $R_\agent$ is reflexive and transitive. We now have to show that each $\evidence_\agent$ satisfies the conditions of Definition \ref{def:evidence function for LPLTL}. We only show here the $(monotonicity)$ condition for $\evidence_\agent$. The proof for other conditions are easy.
	
	Suppose that $X, Y \in S$, $X R_\agent Y$, and $\phi \in \evidence_\agent (X,t)$. We have to show that $\phi \in \evidence_\agent (Y,t)$. Let $Y = \overline{\Delta}$, for $\Delta \in \MCS$. We have to show that $\jbox{t}_\agent \phi \in \Delta$. Since $X R_\agent Y$, there exists $\Gamma \in \MCS$ such that $X = \overline{\Gamma}$ and  
	$\{ \phi \mid  \jbox{t}_\agent \phi \in \Gamma   \text{ for some $t$}\} \subseteq \Delta$. Since $\phi \in \evidence_\agent (X,t)$, we get $\jbox{t}_\agent \phi \in \Gamma$, and thus $\jbox{!t}_\agent \jbox{t}_\agent \phi \in \Gamma$. Therefore, $\jbox{t}_\agent \phi \in \Delta$, as desired.
\end{proof}

\begin{lemma}[Truth Lemma]\label{lem:Truth Lemma LPLTL}
	Let $\mathsf{L} = \LPLTLp$, let $\CS$ be a constant specification for ${\sf L}$, and let
	$
	\system = (\runs, S, R_1,\ldots,R_\numberofagents, \evidence_1\ldots,\evidence_\numberofagents, \valuation)
	$ 
	be the $\chi$-canonical interpreted system for~$\LPLTLp_\CS$. For every formula $\psi \in  \Subf^+(\chi)$, every run~$r$ in $\runs$, and every $n \in \N$ we have:
	\[
	(\system, r, n) \models \psi 
	\quad\text{if{f}}\quad
	\psi \in r(n).
	\]
\end{lemma}
\begin{proof}
	As usual, the proof is by induction on the structure of $\psi$.\linebreak Let $\RO{r(i-1)}{r(i)}{\overline{\Gamma_{i-1}}}{\overline{\Delta_{i}}}$, for $i \geq 1$. We only show the following cases:
	\begin{itemize}
		\item $\psi = \lnext \phi$.
		
		$(\system, r, n) \models \lnext \phi$ if{f} $(\system, r, n+1) \models \phi$, by the induction hypothesis, if{f} $\phi \in r(n+1)$ if{f} $\phi \in \Delta_{n+1}$, by Lemma \ref{lem:R-next}, if{f} $\lnext \phi \in \Gamma_{n}$ if{f} $\lnext \phi \in r(n)$.
		
		\item $\psi = \wprevious \phi$.
		
		$(\Rightarrow)$ Suppose that $(\system, r, n) \models \wprevious \phi$ and $\wprevious \phi \not\in r(n)$.  Then $n=0$ or $(\system, r, n-1) \models  \phi$.
		
		\begin{itemize}
			\item Suppose $n=0$. Since $r(0)$ is initial, $\wprevious \bot \in r(0)$. Since  $\vdash_{\mathsf{L}_\CS} \wprevious \bot \rightarrow \wprevious \phi$, by Lemma \ref{lem: Facts about MCS-chi}, we get $\wprevious\phi \in r(0)$. The latter clearly contradicts the assumption  $\wprevious \phi \not\in r(0)$.
			
			\item Suppose $n>0$ and $(\system, r, n-1) \models  \phi$. Hence by the induction hypothesis $\phi \in r(n-1)$. Since $\RO{r(n-1)}{r(n)}{\overline{\Gamma_{n-1}}}{\overline{\Delta_{n}}}$, we have $\phi \in \Gamma_{n-1}$. Hence, by Lemma \ref{lem:R-next}, $\sprevious \phi \in \Delta_n$. By axiom $\swprevax$, we get   $\wprevious \phi \in \Delta_n$, which is a contradiction.
		\end{itemize}
		
		$(\Leftarrow)$ Suppose $\wprevious \phi \in r(n)$ and  $n>0$.  Since $\RO{r(n-1)}{r(n)}{\overline{\Gamma_{n-1}}}{\overline{\Delta_{n}}}$, we have $\wprevious \phi \in \Delta_{n}$, by Lemma \ref{lem:R-next}, we get $\phi \in \Gamma_{n-1}$. By the induction hypothesis, $(\system, r, n-1) \models  \phi$, and hence $(\system, r, n) \models \wprevious \phi$. 

		\item $\psi = \psi_1 \luntil \psi_2$.
		
		$(\Rightarrow)$ If $(\system, r, n) \models  \psi_1 \luntil \psi_2$, then $(\system, r, m) \models  \psi_2$ for some $m \geq n$, and $(\system, r, k) \models  \psi_1$ for all $k$ with $n \leq k < m$. By the induction hypothesis we get $\psi_2 \in r(m)$, and $\psi_1 \in r(k)$ for all~$k$ with $n \leq k < m$. We have to show $\psi_1 \luntil \psi_2 \in r(n)$, which follows by induction on $m$ as follows:
		\begin{itemize}
			\item
			Base case $m=n$. Since $\psi_2 \in r(n) = r(m)$ and $\vdash_{\mathsf{L}_\CS} \psi_2 \rightarrow (\psi_1 \luntil \psi_2)$, by Lemma \ref{lem: Facts about MCS-chi}, we obtain $\psi_1 \luntil \psi_2 \in r(n)$. 
			
			\item
			Suppose $m > n$.  It follows from the induction hypothesis that $\psi_1 \luntil \psi_2 \in r(n+1)$. Since $\RO{r(n)}{r(n+1)}{ \overline{\Gamma_n}}{ \overline{\Delta_{n+1}}}$, and hence $\psi_1 \luntil \psi_2 \in \Delta_{n+1}$. Thus, by Lemma \ref{lem:R-next}, $\lnext (\psi_1 \luntil \psi_2) \in \Gamma_n$. Now suppose towards a contradiction that $\psi_1 \luntil \psi_2 \not\in \Gamma_n$. Hence $\neg(\psi_1 \luntil \psi_2) \in \Gamma_n$. By axiom $(\luntil 2)$, 
			\[\vdash_{\mathsf{L}_\CS}  \neg (\psi_1 \luntil \psi_2) \rightarrow [\neg \psi_2 \wedge (\neg \psi_1 \vee \neg \lnext (\psi_1 \luntil \psi_2))], \]
			and thus 
			\[\vdash_{\mathsf{L}_\CS}  \neg (\psi_1 \luntil \psi_2) \wedge \psi_1 \rightarrow \neg \lnext (\psi_1 \luntil \psi_2), \]
			Thus, $\neg \lnext (\psi_1 \luntil \psi_2) \in \Gamma_n$, which is a contradiction. Thus, $\psi_1 \luntil \psi_2 \in \Gamma_n$ and hence $\psi_1 \luntil \psi_2  \in r(n)$.
			
		\end{itemize}
		
		$(\Leftarrow)$ If $\psi_1 \luntil \psi_2 \in r(n)$, then since $(r(0), r(1), \ldots, r(n),r(n+1),\ldots)$ is an acceptable sequence there exists $m \geq n$ such that $\psi_2 \in r(m)$, and $\psi_1 \in r(k)$ for all~$k$ with $n \leq k < m$. By the induction hypothesis we obtain 
		$(\system, r, m) \models  \psi_2$, and $(\system, r, k) \models  \psi_1$ for all $k$ with $n \leq k < m$. Thus 
		$(\system, r, n) \models  \psi_1 \luntil \psi_2$.
		
		\item $\psi = \psi_1 \lsince \psi_2$. 
		
		$(\Rightarrow)$ If $(\system, r, n) \models \psi_1 \lsince \psi_2$, then $(\system, r, m) \models \psi_2$ for some $m \leq n$, and $(\system, r, k) \models \psi_1$ for all $k$ with $m < k \leq n$. By the induction hypothesis, $\psi_2 \in r(m)$, and $\psi_1 \in r(k)$ for all $k$ with $m < k \leq n$. We want to show that $\psi_1 \lsince \psi_2 \in r(n)$. We prove it by induction on $m$ as follows.
		\begin{itemize}
			\item 		Base case $m=n$. Since $\psi_2 \in r(n)=r(m)$ and $\vdash_{\mathsf{L}_\CS}  \psi_2 \rightarrow (\psi_1 \lsince \psi_2)$, we obtain $\psi_1 \lsince \psi_2 \in r(n)$. 
			
			\item 		Suppose $m < n$. Since $\RO{r(n-1)}{r(n)}{\overline{\Gamma_{n-1}}}{\overline{\Delta_{n}}}$, it follows from the induction hypothesis that $\psi_1 \lsince \psi_2 \in r(n-1)$, and hence $\psi_1 \lsince \psi_2 \in \Gamma_{n-1}$. Thus, by Lemma \ref{lem:R-next}, $\sprevious (\psi_1 \lsince \psi_2) \in \Delta_n$. Now suppose towards a contradiction that $\psi_1 \lsince \psi_2 \not\in \Delta_n$. Hence $\neg(\psi_1 \lsince \psi_2) \in \Delta_n$. By axiom  $\stwoax$,
			\[\vdash_{\mathsf{L}_\CS}  \neg (\psi_1 \lsince \psi_2) \rightarrow [\neg \psi_2 \wedge (\neg \psi_1 \vee \neg \sprevious (\psi_1 \lsince \psi_2))], \]
			and thus 
			\[\vdash_{\mathsf{L}_\CS}  \neg (\psi_1 \lsince \psi_2) \wedge \psi_1 \rightarrow \neg \sprevious (\psi_1 \lsince \psi_2), \]
			Thus, $\neg \sprevious (\psi_1 \lsince \psi_2) \in \Delta_n$, which is a contradiction.
			
			$(\Leftarrow)$ Suppose $\psi_1 \lsince \psi_2 \in r(n)$. By Lemma \ref{lem: MCS temporal properties}, there is $m \leq n$ such that $\psi_2 \in r(m)$, and $\psi_1 \in r(k)$ for all $k$ with $m < k \leq n$. By the induction hypothesis, $(\system, r, m) \models \psi_2$ and $(\system, r, k) \models \psi_1$ for all $k$ with $m < k \leq n$, and thus $(\system, r, n) \models \psi_1 \lsince \psi_2$ as desired. 
		\end{itemize}

		\item $\psi = \jbox{t}_\agent \phi$.

		$(\Rightarrow)$ If $(\system, r, n) \models \jbox{t}_\agent \phi$, then $\phi \in \evidence_\agent (r(n),t)$. Thus, by the definition of $\evidence_\agent$, $\jbox{t}_\agent \phi \in \Gamma$, where $r(n) = \overline{\Gamma}$, and hence $\jbox{t}_\agent \phi \in r(n)= \overline{\Gamma_n}$. 
		
		$(\Leftarrow)$ If $\jbox{t}_\agent \phi \in r(n)$, then $\jbox{t}_\agent \phi \in \bigcap_{r(n) = \overline{\Gamma}} \Gamma$. Thus, $\phi \in \evidence_\agent (r(n),t)$. 
		Now suppose that $r(n) R_\agent  r'(n')$ and let $r'(n') = \overline{\Delta}$, for some $\Delta \in \MCS$. By the definition of $R_\agent$, there is $\Gamma \in \MCS$ such that $r(n) = \overline{\Gamma}$ and  
		$\{ \phi \ |\  \jbox{t}_\agent \phi \in \Gamma   \text{ for some $t$}\} \subseteq \Delta$. We find $\phi \in \Delta$. By the induction hypothesis, we get $(\system, r', n') \models \phi$. Since $r'$ and $n'$ were arbitrary, we conclude that $(\system, r, n) \models \jbox{t}_\agent \phi$. 
	\end{itemize}
\end{proof}

\begin{theorem}[Completeness]\label{thm:Completeness-interpreted systems}
	Let ${\sf L} = \LPLTLp$ and let $\CS$ be a constant specification for ${\sf L}$. For each formula $\chi$,
	\[
	\models_{{\sf L}_\CS}  \chi  \quad\text{implies}\quad  \vdash_{{\sf L}_\CS} \chi.
	\]
\end{theorem}
\begin{proof}
	Suppose that $\not\vdash_{{\sf L}_\CS} \chi$. Thus, $\{ \neg \chi\}$ is an ${\sf L}_\CS$-consistent set.
	Therefore, there exists $\Gamma \in \MCS$ with $\neg \chi \in \Gamma$. Let $\overline{\Gamma} = \Gamma \cap  \Subf^+(\chi)$.
	By Corollary \ref{cor:acceptable sequence starts with Gamma}, there is an acceptable sequence containing $\overline{\Gamma}$, say $(X_{0},X_1,\ldots)$ where $\overline{\Gamma} = X_n$, for some $n \geq 0$. Define the run $r$ as follows $r(i) \colonequals X_i$. The run $r$ is in the system $\runs$ of the $\chi$-canonical model $\system$ for ${\sf L}_\CS$. 
	Since $\chi \not \in r(n)$, by the Truth Lemma, $(\system, r, n) \not\models \chi$. Therefore, $\not\models_{{\sf L}_\CS} \chi$.  
\end{proof}

\begin{theorem}[Completeness]\label{thm:Weak Completeness-interpreted systems}
	Let ${\sf L} = \LPLTLp$ and let $\CS$ be a constant specification for ${\sf L}$. For each formula $\chi$ and finite set of formulas $T$,
	\[
	T \models_{{\sf L}_\CS}  \chi  \quad\text{implies}\quad  T \vdash_{{\sf L}_\CS} \chi.
	\]
\end{theorem}
\begin{proof}
	Suppose that $T \not \vdash_{{\sf L}_\CS} \chi$. Thus, $\not \vdash_{{\sf L}_\CS} \bigwedge T \to \chi$. By Theorem \ref{thm:Completeness-interpreted systems}, there is an interpreted system $\system = (\runs, \ldots)$, $r \in \runs$, and $n \in \N$ such that $(\system, r, n) \models \bigwedge T$ and $(\system, r, n) \not\models \chi$. Therefore, $T \not \models_{{\sf L}_\CS} \chi$.  
\end{proof}

\section{Another semantics for $\LPLTLp$}
\label{sec:Mkrtychev models aditional principles}

In this section we present another semantics based on Mkrtychev models \cite{Mkr97LFCS} for $\LPLTLp$. These models are indeed the interpreted systems with singleton system of runs.

\begin{definition}\label{def:M-models LPLTLp}
 An $\LPLTLp_\CS$-model is a tuple
$
\M = (r,S,\evidence_1\ldots,\evidence_\numberofagents, \valuation)
$ 
where
\begin{enumerate}
	\item  $S$ is a non-empty set of states;
	\item $r : \mathbb{N} \to S$ is a run on $S$;
	\item $\evidence_\agent$ is a $\CS$-evidence function for agent~$\agent$, for $1 \leq \agent \leq \numberofagents$, that satisfies conditions 2--5 of Definition \ref{def:evidence function for LPLTL};
	\item $\valuation: S \to \powerset(\Prop)$ is a valuation.
\end{enumerate}
\end{definition}

Given an $\LPLTLp_\CS$-model $\M = (r,S,\evidence_1\ldots,\evidence_\numberofagents, \valuation)$
and  $n \in \N$, we define truth of a formula $\phi$ in $\M$ at state $r(n)$ inductively by 
\begin{align*}
(\M, r(n)) &\entails P \text{ iff } P \in \valuation(r(n)) \, ,\\
(\M, r(n)) &\not\entails \lfalse \, ,\\
(\M, r(n)) &\entails \phi \limplies \psi \text{ iff } (\M, r(n)) \not\entails \phi \text{ or } (\M, r(n)) \entails \psi \, ,\\
(\M, r(n)) &\entails \wprevious \phi \text{ iff $n=0$ or } (\M, r(n-1)) \entails \phi \, ,\\
(\M, r(n)) &\entails \lnext \phi \text{ iff } (\M, r(n+1)) \entails \phi \, ,\\
(\M, r(n)) &\entails \phi \lsince \psi \text{ iff there is some } m \leq n \text{ such that } (\M,  r(m)) \entails \psi \\ & \qquad\qquad \text{ and } (\M, r(k)) \entails \phi \text{ for all $k$ with } m < k \leq n \, ,\\
(\M, r(n)) &\entails \phi \luntil \psi \text{ iff there is some } m \geq n \text{ such that } (\M, r(m)) \entails \psi \\ & \qquad\qquad \text{ and } (\M, r(k)) \entails \phi \text{ for all $k$ with } n \leq k < m \, ,\\
(\M, r(n)) &\entails \jbox{t}_\agent \phi \text{ iff }   \phi \in \evidence_\agent(r(n),t ) \text { and } (\M, r(n)) \entails \phi.
\end{align*}

We write $\M \entails \phi$ if $(\M, r(n)) \entails \phi$ for all $n \in \N$. 

From the above definitions it follows that: 
\begin{align*}
(\M, r(n)) &\entails \leventually \phi \text{ iff } (\M, r(m)) \entails \phi \text{ for some } m\geq n \, ,\\
(\M, r(n)) &\entails \lalways \phi \text{ iff } (\M, r(m)) \entails \phi \text{ for all } m\geq n \, ,\\
(\M, r(n)) &\entails \lonce \phi \text{ iff } (\M, r(m)) \entails \phi \text{ for some } m \leq n \, ,\\
(\M, r(n)) &\entails \lsofar \phi \text{ iff } (\M, r(m)) \entails \phi \text{ for all } m \leq n \, ,\\
(\M, r(n)) &\entails \sprevious \phi \text{ iff $n>0$ and } (\M, r(n-1)) \entails \phi  \, ,\\
\end{align*}

\begin{definition}
	Given a set of formulas $\Gamma$ and a formula $\phi$, the (local) consequence relation is defined as follows: $\Gamma \Vdash_\CS \phi$ iff for all 
	$\LPLTLp_\CS$-models $\M = (r , \ldots)$, and for all $ n \in \N$, if $(\M, r(n)) \entails \psi$ for all $\psi \in \Gamma$, then $(\M, r(n)) \entails \phi$.
\end{definition}

\begin{remark}
	As in Remark \ref{rem:time is n-interpreted systems}, if we let ``$\ltime=m$" abbreviate $\sprevious^m \wprevious \bot$ and $\ltrue_m (\phi)$ abbreviate $\ltime = m \to \phi$, then we have
	\[
	\M \models \ltrue_m (\phi) \text{ iff } \M, m \models \phi.
	\]
\end{remark}


\begin{theorem}[Soundness]
		For each formula $\phi$ and finite set of formulas $\Gamma$,
	\[
	\Gamma\vdash_\CS \phi  \quad\text{implies}\quad  \Gamma \Vdash_\CS  \phi.
	\]
\end{theorem}


\begin{definition}\label{def:canonical M-model} 
	Let  $(X_0, X_1, \ldots)$ be an acceptable sequence of elements of $\MCS_\chi$ for $\LPLTLp_\CS$. The $\chi$-canonical  model~$\M = (r,S,\evidence_1,\ldots,\evidence_\numberofagents, \valuation)$ for $\CS$
	with respect to $(X_0, X_1, \ldots)$
	is  defined as follows:
	\begin{enumerate}
		
		\item $S \colonequals \{X_0, X_1, \ldots\}$.
		\item $r(n) := X_n$.
		\item $\evidence_\agent(X_n,t) \colonequals \{ \phi \mid \jbox{t}_\agent \phi \in \bigcap \{ \Gamma \in \MCS \mid X = \overline{\Gamma} \}$.
		\item $\valuation(X_n) \colonequals \Prop\cap X_n$.
		
	\end{enumerate}
	
\end{definition}

%
%
%
%
%

\begin{lemma}
	The $\chi$-canonical model $\M = (r,S,\evidence_1,\ldots,\evidence_\numberofagents, \valuation)$ for $\CS$  with respect to an acceptable sequence $(X_0, X_1, \ldots)$ is an $\LPLTLp_\CS$-model.
\end{lemma}
\begin{proof}
	Similar to the proof of Lemma \ref{lem: canonical interpreted model is a model}.\qed
\end{proof}

\begin{lemma}[Truth Lemma]
	Let $\M = (r,S,\evidence_1,\ldots,\evidence_\numberofagents, \valuation)$ be the $\chi$-canonical model for $\CS$ with respect to an acceptable sequence $(X_0, X_1, \ldots)$. For every formula $\psi \in \Sub^+(\chi)$, and every $n \in \N$ we have:
	\[ 
	(\M, r(n)) \models \psi 
	\quad\text{if{f}}\quad
	\psi \in r(n). 
	\]
\end{lemma}
\begin{proof} 
As usual, the proof is by induction on the structure of $\psi$. Let \linebreak $\RO{r(i-1)}{r(i)}{\overline{\Gamma_{i-1}}}{\overline{\Delta_{i}}}$, for $i \geq 1$.
Since the proof is similar to the proof of Lemma \ref{lem:Truth Lemma LPLTL}, we only show the following case:
\begin{itemize}
	\item $\psi = \jbox{t}_\agent \phi$.
	
	$(\Rightarrow)$ If $(\M, r(n)) \models \jbox{t}_\agent \phi$, then $\phi \in \evidence_\agent (r(n),t)$. Thus, by the definition of $\evidence_\agent$, $\jbox{t}_\agent \phi \in \Gamma$, where $r(n) = \overline{\Gamma}$, and hence $\jbox{t}_\agent \phi \in r(n)= \overline{\Gamma_n}$. 
	
	$(\Leftarrow)$ If $\jbox{t}_\agent \phi \in r(n)$, then by axiom $\refax$, we have  $\phi \in r(n)$. By the induction hypothesis, we get $(\M, r(n)) \models \phi$. We conclude $(\M, r(n)) \models \jbox{t}_\agent \phi$.
	\qed
\end{itemize}
\end{proof}

\begin{theorem}[Completeness]\label{thm:Completeness M-models}
	\label{thm:completeness LPLTLp M-models}
	If	$\M\entails \phi$ for all $\LPLTLp_\CS$-models $\M$, then  $\proves_\CS \phi$.
\end{theorem}
\begin{proof}
	Suppose that $\LPLTLp_\CS \not\proves \phi$. Thus, $\{ \neg \phi\}$ is a $\CS$-consistent set.
	Therefore, there exists $\Gamma \in \MCS$ with $\neg \phi \in \Gamma$. Let $\overline{\Gamma} = \Gamma \cap \Sub^+(\phi)$.
	By Corollary \ref{cor:acceptable sequence starts with Gamma}, there is an acceptable sequence containing $\overline{\Gamma}$, say $(\overline{\Gamma_{0}},\overline{\Gamma_1}\ldots,\overline{\Gamma_n}, \overline{\Gamma_{n+1}}, \ldots)$ where $n \geq 0$ and  $\overline{\Gamma_n} = \overline{\Gamma}$. 
	Construct the $\phi$-canonical model $\M$ for $\CS$ with respect to this acceptable sequence.
	Since $\phi \not \in \Gamma_n$, by the Truth Lemma, $(\M, r(n)) \not\models \phi$. Therefore, $\M \not\models \phi$.\qed
\end{proof}

\begin{theorem}[Completeness]\label{thm:Weak Completeness M-models}
	For each formula $\phi$ and finite set of formulas $\Gamma$,
	\[
	\Gamma \Vdash_\CS  \phi  \quad\text{implies}\quad  \Gamma \vdash_\CS \phi.
	\]
\end{theorem}
\begin{proof}
	Suppose that $\Gamma \not\vdash_\CS \phi$. Thus, $\not \vdash_\CS \bigwedge \Gamma \to \phi$. By Theorem \ref{thm:Completeness M-models}, there is an $\LPLTLp_\CS$-model $\M = (r, \ldots)$,  and $n \in \N$ such that $(\M, r(n)) \models \bigwedge \Gamma$ and $(\M, r(n)) \not\models \phi$. Therefore, $\Gamma \not \Vdash_\CS \phi$. \qed
\end{proof}
\section{Connecting principles}
\label{sec:Connecting principles}

In $\LPLTLp_\CS$, epistemic and temporal properties do not interact. In this section we study some principles that create a connection between justifications and temporal modalities.
We assume the language for terms to be augmented in the obvious way.

\begin{enumerate}
	\item $\lalways \jbox{t}_\agent \phi \limplies \jbox{\tgeneralize t}_\agent \lalways  \phi$ \hfill \generalizeprinciple
	\item $\jbox{t}_\agent \lalways \phi \limplies \lalways \jbox{\talwaysaccess t}_\agent \phi$ \hfill \alwaysaccessprinciple
	\item $\jbox{t}_\agent \lalways \phi \limplies \jbox{\tnextaccess t}_\agent \lnext \phi$ \hfill \nextaccessprinciple
	\item $\jbox{t}_\agent \lnext \phi \limplies \lnext \jbox{\tnext t}_\agent \phi$ \hfill \nextrightshiftprinciple
	\item $\lnext \jbox{t}_\agent \phi \limplies \jbox{\tprev t}_\agent \lnext \phi$ \hfill \nextleftshiftprinciple
	\vspace*{0.1cm}
	\item $\lsofar \jbox{t}_\agent \phi \limplies \jbox{\thenceforthgeneralize t}_\agent \lsofar  \phi$ \hfill \pastgeneralizeprinciple
	\item $\jbox{t}_\agent \lsofar \phi \limplies \lsofar \jbox{\thenceforthaccess t}_\agent \phi$ \hfill \pastaccessprinciple
	\item $\jbox{t}_\agent \lsofar \phi \limplies \jbox{\tprevaccess t}_\agent \wprevious \phi$ \hfill \wprevaccessprinciple
	\item $\jbox{t}_\agent \wprevious \phi \limplies \wprevious \jbox{\twprevright t}_\agent \phi$ \hfill \wprevrightshiftprinciple
	\item $\jbox{t}_\agent \sprevious \phi \limplies \sprevious \jbox{\tsprevright t}_\agent \phi$ \hfill \sprevrightshiftprinciple
	\item $\sprevious \jbox{t}_\agent \phi \limplies \jbox{\tsprevleft t}_\agent \sprevious \phi$ \hfill \sprevleftshiftprinciple
\end{enumerate}

Principles 1--5 were first proposed by Bucheli in \cite{Bucheli15} from which the name of the axioms are also taken.\footnote{The principle $\jbox{t}_\agent \lnext \phi \limplies \lnext \jbox{\tnext t}_\agent \phi$ is called $\nextaccessprinciple$ in \cite{BucheliGhariStuder2017}.}
A few remarks on these principles are in order:

\begin{description}
	\item[\generalizeprinciple] This principle says that if you have a fixed piece of evidence that always supports a proposition, then you have evidence that this proposition is always true. The term operator $\tgeneralize$ converts permanent evidence for a proposition to evidence for knowing that this proposition is always true.
		
	\item[\alwaysaccessprinciple] This principle says that if you have evidence that a proposition is always true, then at every point in time you are able to access this information. The term operator $\talwaysaccess$ makes the evidence accessible in every future point in time. This principle is a counterpart of the axiom $K_i \lalways \phi \to \lalways K_i \phi$ in the logics of knowledge and time, which is valid in the interpreted systems with \emph{perfect recall} (where an agent retains the knowledge of previous times), but does not characterize it,  see \cite{HvdMV04}.
	
	\item[\nextaccessprinciple] This principle is similar to the valid formula $\lalways \phi \to \lnext \phi$ augmented by justifications. In fact, if you have evidence that a proposition is always true, then you have evidence that it is true tomorrow, and the term operator $\tnextaccess$ constructs such an evidence. 

	\item[\nextrightshiftprinciple] This principle says that agents do not forget evidence once they have gathered it and can ``take it with them''. The term operator $\tnext$ carries evidence through time. This principle is a counterpart of the axiom $K_i \lnext \phi \to \lnext K_i \phi$ in the logics of knowledge and time, which characterizes the  \emph{synchronous} systems (where each agent always knows the time) with perfect recall, see \cite{HvdMV04}.
	
	\item[\nextleftshiftprinciple] This principle implies some form of conditional prediction. The term operator $\tprev$ predicts future evidence for knowledge. This principle is a counterpart of the axiom $\lnext K_i \phi \to K_i \lnext \phi$ in the logics of knowledge and time, which characterizes the synchronous  systems with \emph{no learning} (where an agent's knowledge can not increase over time), see \cite{HvdMV04}.
\end{description}

The connecting principles involving past operators are the dual of those involving future operators, and thus the meaning of the term operators with subscript $P$ can be guessed straightforwardly. 

Given a logic ${\sf L}$ and a set of axioms $Ax$, by ${\sf L}(Ax)$ we denote the result of adding axioms  from $Ax$ to logic ${\sf L}$. In the rest of this section $Ax$ is an arbitrary set of the above connecting principles, i.e.

\begin{gather*}
Ax \subseteq \{ \generalizeprinciple, \alwaysaccessprinciple, \nextaccessprinciple, \nextrightshiftprinciple, \nextleftshiftprinciple, 
\pastgeneralizeprinciple,\\ \pastaccessprinciple, \wprevaccessprinciple, \wprevrightshiftprinciple, \sprevrightshiftprinciple, \sprevleftshiftprinciple \}.
\end{gather*}

In the following sections we introduce new axioms and allow $Ax$ to include the new axioms as well.

%
%
%
%

Let us show that a version of $\nextaccessprinciple$ is derivable from $\alwaysaccessprinciple$ and $\nextleftshiftprinciple$.

\begin{lemma}
	Let $Ax$ contains axioms $\alwaysaccessprinciple$ and $\nextleftshiftprinciple$. For every agent~$\agent$, formula~$\phi$ and term~$t$ there is a term~$s(t)$ such that
	\[
	\vdash_{\LPLTLp(Ax)_\emptyset} \jbox{t}_\agent \lalways \phi \limplies \jbox{s(t)}_\agent \lnext \phi.
	\]
\end{lemma}

\begin{proof}
	Construct the following proof in $\LPLTLp(Ax)_\emptyset$ where $Ax$ contains axioms $\alwaysaccessprinciple$ and $\nextleftshiftprinciple$.
	\begin{enumerate}
		\item $\jbox{t}_\agent \lalways \phi \limplies \lalways \jbox{\talwaysaccess t}_\agent \phi$, \hfill instance of axiom \alwaysaccessprinciple
		
		\item $\lalways \jbox{\talwaysaccess t}_\agent \phi \limplies \lnext \jbox{\talwaysaccess t}_\agent \phi$, \hfill Lemma \ref{lem:mix:1} item 2
		
		\item $\lnext \jbox{\talwaysaccess t}_\agent \phi \limplies \jbox{\tprev \talwaysaccess t}_\agent \lnext \phi$, \hfill instance of axiom \nextleftshiftprinciple
		
		\item $\jbox{t}_\agent \lalways \phi \limplies \jbox{\tprev\talwaysaccess t}_\agent \lnext \phi$. \hfill from 1--3 by propositional reasoning
	\end{enumerate}
Finally put $s(t) := \tprev\talwaysaccess t$.\qed

\end{proof}

\subsection{Semantics}

Now we present a semantics for $\LPLTLp(Ax)$ based on Mkrtychev models. In the next section these models will be extended to interpreted systems.

\begin{definition}\label{def: evidence functions M-models aditional principles}
Let $\CS$ be a constant specification for $\LPLTLp(Ax)$. An \linebreak $\LPLTLp(Ax)_\CS$-model   is a tuple $
\M = (r, S, \evidence_1\ldots,\evidence_\numberofagents, \valuation)
$
where $S$ is a non-empty set of states, $r$ is a run on $S$, $\valuation$ is a valuation, and the $\CS$-evidence functions $\evidence_1\ldots,\evidence_\numberofagents$ should
	satisfy conditions 2--5 of Definition \ref{def:evidence function for LPLTL} and the following conditions depending on axioms in $Ax$.
	For all $n \in \N$, all terms $s,t \in \Terms$, all formulas $\phi,\psi \in \Formulae$, and all $i \in \Ag$:
	\begin{enumerate}
		\item if $\phi \in \evidence_\agent(r(m), t)$ for all $m \geq n$, then $\lalways \phi \in \evidence_i(r(n), \tgeneralize t)$. \hfill
		$\generalizeevidence$
				
		\item if $\lalways\phi \in \evidence_\agent(r(n),t)$, then $\phi \in \evidence_\agent(r(m), \talwaysaccess t)$ for all $m \geq n$. \hfill $\alwaysaccessevidence$
		
		\item if $\lalways\phi \in \evidence_\agent(r(n),t)$, then $\lnext\phi \in \evidence_\agent(r(n), \tnextaccess t)$.  \hfill $\nextaccessevidence$
			
		\item if $\lnext \phi \in \evidence_\agent(r(n),t)$, then $\phi \in \evidence_\agent(r(n+1), \tnext t)$. 
		\hfill $\nextrightshiftevidence$
		
		\item if $\phi \in \evidence_\agent(r(n+1), t)$, then $\lnext \phi \in \evidence_\agent(r(n), \tprev t)$.
		\hfill $\nextleftshiftevidence$
		
		\vspace*{0.1cm}
		\item if $\phi \in \evidence_\agent(r(m), t)$ for all $m \leq n$, then $\lsofar \phi \in \evidence_i(r(n), \thenceforthgeneralize t)$. \hfill
		$\pastgeneralizeevidence$
		
		\item if $\lsofar \phi \in \evidence_\agent(r(n),t)$, then $\phi \in \evidence_\agent(r(m), \thenceforthaccess t)$ for all $m \leq n$. \hfill $\pastaccessevidence$
		
		\item if $\lsofar \phi \in \evidence_\agent(r(n),t)$, then $\wprevious \phi \in \evidence_\agent(r(n), \tprevaccess t)$.  \hfill $\wprevaccessevidence$
		
		\item $\wprevious \phi \in \evidence_\agent(r(n),t)$ and $n>0$, then $\phi \in \evidence_\agent(r(n-1), \twprevright t)$. \hfill $\wprevrightshiftevidence$

		\item $\sprevious \phi \in \evidence_\agent(r(n),t)$ and $n>0$, then $\phi \in \evidence_\agent(r(n-1), \tsprevright t)$.  \hfill $\sprevrightshiftevidence$
		
		\item $\phi \in \evidence_\agent(r(n-1),t)$, then $\sprevious \phi \in \evidence_\agent(r(n), \tsprevleft t)$. \hfill $\sprevleftshiftevidence$
	\end{enumerate}
\end{definition}

\begin{lemma}\label{lem: MCS temporal properties}
Let $(\overline{\Gamma_0}, \overline{\Gamma_1}, \ldots)$ be an acceptable sequence of elements of $\MCS_\chi$. 

\begin{enumerate}
	\item If $\lalways \phi \in \Gamma_n$, then $\phi \in \Gamma_m$ for all $m \geq n$.
	
	\item If $\leventually \phi \in \overline{\Gamma_n}$, then $\phi \in \overline{\Gamma_{m}}$ for some $m \geq n$.
	
	\item If $\lonce \phi \in \Gamma_n$, then $\phi \in \Gamma_{m}$ for some $m \leq n$.
	

	\item  $\lsofar \phi \in \Gamma_n$ iff $\phi \in \Gamma_m$ for all $m \leq n$.
\end{enumerate}
\end{lemma}
\begin{proof}
\begin{enumerate}
	\item Suppose $\lalways \phi \in \Gamma_n$ and there is $m \geq n$ such that $\phi \not \in \Gamma_{m}$. If $m=n$, then by Lemma \ref{lem:mix:1} we get $\phi \in \Gamma_n$, which is a contradiction. Now suppose $m > n$. Then, by Lemma \ref{lem:mix:1}, $\lnext \lalways \phi \in \Gamma_n$. By Lemma \ref{lem:R-next}, $\lalways \phi \in \Gamma_{n+1}$.  By repeating the argument, we get $\lalways \phi \in \Gamma_m$, and hence $\phi \in \Gamma_m$  which is a contradiction.
	
	\item Suppose $\leventually \phi \in \overline{\Gamma_n}$. Then $ \top \luntil  \phi \in \overline{\Gamma_n}$. Since $(\overline{\Gamma_0}, \overline{\Gamma_1}, \ldots)$ is an acceptable sequence, there exists $m \geq n$ such that $\phi \in \overline{\Gamma_m}$.
	
	\item Suppose $\lonce \phi \in \Gamma_n$. Then $\top \lsince \phi \in \Gamma_{n}$. Then either $\phi \in \Gamma_n$ or $\top \wedge \sprevious (\top \lsince \phi) \in \Gamma_n$. In the former case, we are done. In the latter case, by Lemma \ref{lem:R-next} we have $\top \lsince \phi \in \Gamma_{n-1}$. Again from $\top \lsince \phi \in \Gamma_{n-1}$ it follows that either $\phi \in \Gamma_{n-1}$ or $\top \wedge \sprevious (\top \lsince \phi) \in \Gamma_{n-1}$. In the former case, we are done. In the latter case, by Lemma \ref{lem:R-next} we have $\top \lsince \phi \in \Gamma_{n-2}$. By repeating this argument, we finally get  either $\phi \in \Gamma_0$ or $\sprevious (\top \lsince \phi) \in \Gamma_0$. In the former case, we are done. In the latter case, by Lemma \ref{lem:initial state does not contain strong previous formulas}, we get a contradiction since $\overline{\Gamma_0}$ is initial. 
	
	\item Suppose $\lsofar \phi \in \Gamma_n$ and there is $m \leq n$ such that $\phi \not \in \Gamma_{m}$. If $m=n$, then by Lemma \ref{lem:mix:1} we get $\phi \in \Gamma_n$, which is a contradiction. Now suppose $m < n$. Then, by Lemma \ref{lem:mix:1}, $\wprevious \lsofar \phi \in \Gamma_n$. By Lemma \ref{lem:R-next}, $\lsofar \phi \in \Gamma_{n-1}$.  By repeating the argument, we get $\lsofar \phi \in \Gamma_m$, and hence $\phi \in \Gamma_m$  which is a contradiction.
	
	For the converse suppose that $\phi \in \Gamma_m$ for all $m \leq n$, and $\lsofar \phi \not \in \Gamma_n$. Thus $\neg \lsofar \phi \in \Gamma_{n}$, and hence $\lonce \neg \phi \in \Gamma_{n}$. By clause 3 above, $\neg \phi \in \Gamma_m$ for some $m \leq n$, which would contradict the assumption.\qed
\end{enumerate}
\end{proof}

\begin{theorem}[Soundness and completeness]\label{thm:completeness M-models additional principles}
	Let  $\CS$ be a constant specification for $\LPLTLp(Ax)$. 
	\begin{enumerate}
		\item Suppose that $\generalizeprinciple \in Ax$. If $\phi$  is provable in $\LPLTLp(Ax)_\CS$, then $\M\entails \phi$ for all $\LPLTLp(Ax)_\CS$-models $\M$.
		
		\item Suppose that $\generalizeprinciple \not \in Ax$. $\phi$  is provable in $\LPLTLp(Ax)_\CS$ iff $\M\entails \phi$ for all $\LPLTLp(Ax)_\CS$-models $\M$.
	\end{enumerate}
	
\end{theorem}
\begin{proof}
The proof of soundness of $\LPLTLp(Ax)$, for arbitrary $Ax$, is straightforward. 

Now suppose $\generalizeprinciple$ is not in $Ax$. The proof of completeness of $\LPLTLp(Ax)$ is similar to the proof of Theorem \ref{thm:completeness LPLTLp M-models} by constructing a canonical model. Let $\M = (r,\evidence_1,\ldots,\evidence_\numberofagents, \valuation)$ be the $\chi$-canonical  model for $\CS$
with respect to an acceptable sequence~$(\overline{\Gamma_0}, \overline{\Gamma_1}, \ldots)$ for $\LPLTLp(Ax)_\CS$.	  Truth Lemma can be proved as before. The only new part is to show that $\M$  is an $\LPLTLp(Ax)_\CS$-model.  We only check the details for $\alwaysaccessprinciple$ and $\pastgeneralizeprinciple$, the case of other principles are straightforward.

Let us show that the $\chi$-canonical model of  $\LPLTLp(Ax)$, where $\alwaysaccessprinciple$ is in $Ax$, satisfies the $\alwaysaccessevidence$ condition of Definition \ref{def: evidence functions M-models aditional principles}. Suppose $\lalways \phi \in \evidence_\agent(r(n), t)$. We want to show that $\phi \in \evidence_i(r(m), \talwaysaccess t)$ for all $m \geq n$. It is enough to show that $ \jbox{\talwaysaccess t}_i \phi \in \Gamma_m$ for all $m \geq n$. From Definition \ref{def:canonical M-model}, we get $\jbox{t}_i \lalways \phi \in \Gamma_n$. From axiom $\alwaysaccessprinciple$, we get $ \lalways \jbox{\talwaysaccess t}_i \phi \in \Gamma_n$. Thus, by Lemma \ref{lem: MCS temporal properties}, we have $ \jbox{\talwaysaccess t}_i \phi \in \Gamma_m$ for all $m \geq n$.

Let us now show that the $\chi$-canonical model of  $\LPLTLp(Ax)$, where $\pastgeneralizeprinciple$ is in $Ax$, satisfies the $\pastgeneralizeevidence$ condition of Definition \ref{def: evidence functions M-models aditional principles}. Suppose $\phi \in \evidence_\agent(r(m), t)$ for all $m \leq n$, and thus $\jbox{t}_i \phi \in \Gamma_{m}$ for all $m \leq n$. By Lemma \ref{lem: MCS temporal properties}, we get $\lsofar \jbox{ t}_i \phi \in \Gamma_{n}$. From axiom $\pastgeneralizeprinciple$, we have $\jbox{\thenceforthgeneralize t}_i \lsofar \phi \in \Gamma_{n}$. Thus, $\lsofar \phi \in \evidence_i(r(n), \thenceforthgeneralize t)$ as desired.
\end{proof}

We leave the completeness of $\LPLTLp(Ax)$, where $Ax$ contains $\generalizeprinciple$, as an open problem. In Sections \ref{sec:LPLTLp with annotated Appl} and \ref{sec:Completeness for generalize} we achieve the completeness of  logics involving $\generalizeprinciple$ by changing the justification logic part of $\LPLTLp$.

\subsection{$\LPLTLp$ with indexed application operators}
\label{sec:LPLTLp with annotated Appl}

In this section we formalize temporal justification logics with indexed application operators, denoted by $\LPLTL^I$.\footnote{ The indexed application operators were first suggested by Renne \cite{Ren09TARK}.} Terms and formulas of temporal justification logics with indexed application operators are constructed by the following mutual grammar:  
\[
t \coloncolonequals c \mid x \mid \; \tinspect t \mid \;  t + t \mid t \tapp_\phi t \, ,
\]
\[
\phi \coloncolonequals P \mid \lfalse \mid \phi \limplies \phi \mid \lnext \phi \mid \wprevious \phi \mid \phi \luntil \phi \mid \phi \lsince \phi \mid \jbox{t}_\agent\phi \, . 
\]

Axioms and rules of $\LPLTL^I$ are exactly the same as for $\LPLTLp$, except that axiom $\appax$ is replaced by the following  axiom

\begin{itemize}
	\item  $\jbox{t}_\agent (\phi \limplies \psi) \limplies (\jbox{s}_\agent \phi \limplies \jbox{t \tapp_\phi s}_\agent \psi)$.
\end{itemize}

Interpreted systems for $\LPLTL_\CS^I$ and $\LPLTL_\CS^I$-models are defined as in Definitions \ref{def:interpreted sysytems LPLTLp} and \ref{def:M-models LPLTLp} respectively with the difference that condition (application) of Definition \ref{def:evidence function for LPLTL} is replaced by the following condition:

\begin{itemize}
	\item if $\phi \limplies \psi \in \evidence_\agent(w,t)$ and $\phi \in \evidence_\agent(w,s)$, then $\psi \in \evidence_\agent(w, t \tapp_\phi s)$.
\end{itemize}

The notions of $\LPLTL_\CS^I$-validity is defined as usual.
The proof of soundness and completeness theorems for annotated justification logics  with respect to their models is similar to that of $\LPLTLp$.

\begin{theorem}\label{thm:Sound Compl JL^a}
	Let $\CS$ be a constant specification for $\LPLTL^I$. The formula $\phi$ is provable in  $\LPLTL_\CS^I$ if{f} $\M \models \phi$ for all $\LPLTL_\CS^I$-models $\M$.
\end{theorem}

In order to prove completeness of logics involving axiom $\generalizeprinciple$, we need to change the notion of subformula. The following definition is inspired by the work of Marti and Studer \cite{MartiStuder2018}.

\begin{definition}\label{def:Subf}
	The set of subformulas, denoted by $\Subf$,
	is defined by induction on the rank of formulas as follows:
	
	\begin{itemize}
		\setlength\itemsep{0.01cm}
		\item $ \Subf(P) \colonequals \{P\}$, and $ \Subf(\bot) \colonequals \{\bot\}$.
		\item $ \Subf(* \phi) \colonequals \{* \phi\} \cup  \Subf(\phi)$, where $* \in \{ \lnext, \wprevious \}$.
		\item $ \Subf(\phi \star \psi) \colonequals \{\phi \star \psi\} \cup  \Subf(\phi) \cup  \Subf(\psi)$, where $\star \in \{ \to, \lsince, \luntil \}$.
		\vspace*{0.1cm}
		\item[$\circ$] $\Subf(\jbox{x}_\agent \phi) \colonequals \{\jbox{x}_\agent \phi\} \cup \Subf(\phi)$.
		\item[$\circ$] $\Subf(\jbox{c}_\agent \phi) \colonequals \{\jbox{c}_\agent \phi\} \cup \Subf(\phi)$.
		\item[$\circ$] $\Subf(\jbox{t+s}_\agent \phi) \colonequals \{\jbox{t+s}_\agent \phi\} \cup \Subf(\jbox{t}_\agent \phi) \cup \Subf(\jbox{s}_\agent \phi)$.
		\item[$\circ$] $\Subf(\jbox{s \cdot_\phi t}_\agent \psi) \colonequals \{\jbox{s \cdot_\phi t}_\agent  \psi\} \cup \Subf(\jbox{s}_\agent (\phi \to \psi)) \cup \Subf(\jbox{t}_\agent \phi)$.
		\item[$\circ$] $\Subf(\jbox{!t}_\agent \phi) \colonequals \{\jbox{!t}_\agent \phi\} \cup \Subf(\phi)$.
		\vspace*{0.1cm}
		\item $\Subf(\jbox{\tgeneralize t}_\agent \phi) \colonequals \{\jbox{\tgeneralize t}_\agent \phi\} \cup \Subf(\phi)$, $\phi$ is not an $\lalways$-formula.
		\item $\Subf(\jbox{\tgeneralize t}_\agent \lalways \phi) \colonequals \{\jbox{\tgeneralize t }_\agent \lalways\phi\} \cup \Subf(\jbox{t}_\agent \phi)$.
		\vspace*{0.1cm}
		\item[$\circ$] $\Subf(\jbox{\talwaysaccess t}_\agent \phi) \colonequals \{\jbox{\talwaysaccess t }_\agent \phi\} \cup \Subf(\jbox{t}_\agent \lalways \phi)$.
		\vspace*{0.1cm}
		\item $\Subf(\jbox{\tnextaccess t}_\agent \phi) \colonequals \{\jbox{\tnextaccess t}_\agent \phi\} \cup \Subf(\phi)$, $\phi$ is not a $\lnext$-formula.
		\item $\Subf(\jbox{\tnextaccess t}_\agent \lnext \phi) \colonequals \{\jbox{\tnextaccess t }_\agent \lnext \phi\} \cup \Subf(\jbox{t}_\agent \lalways \phi)$.
		\vspace*{0.1cm}
		\item[$\circ$] $\Subf(\jbox{\tnext t}_\agent \phi) \colonequals \{\jbox{\tnext t}_\agent \phi\} \cup \Subf(\jbox{t}_\agent \lnext \phi)$.
		\vspace*{0.1cm}
		\item $\Subf(\jbox{\tprev t}_\agent \phi) \colonequals \{\jbox{\tprev t}_\agent \phi\} \cup \Subf(\phi)$, $\phi$ is not a $\lnext$-formula.
		\item $\Subf(\jbox{\tprev t}_\agent \lnext \phi) \colonequals \{\jbox{\tprev t }_\agent \lnext \phi\} \cup \Subf(\jbox{t}_\agent \phi)$.
		\vspace*{0.1cm}
		\item[$\circ$] $\Subf(\jbox{\thenceforthgeneralize t}_\agent \phi) \colonequals \{\jbox{\thenceforthgeneralize t}_\agent \phi\} \cup \Subf(\phi)$, $\phi$ is not a $\lsofar$-formula.
		\item[$\circ$] $\Subf(\jbox{\thenceforthgeneralize t}_\agent \lsofar \phi) \colonequals \{\jbox{\thenceforthgeneralize t }_\agent \lsofar \phi\} \cup \Subf(\lsofar \jbox{t}_\agent \phi)$.
		\vspace*{0.1cm}
		\item $\Subf(\jbox{\thenceforthaccess t}_\agent \phi) \colonequals \{\jbox{\thenceforthaccess t}_\agent \phi\} \cup \Subf(\jbox{t}_\agent \lsofar \phi)$.
		\vspace*{0.1cm}
		\item[$\circ$] $\Subf(\jbox{\tprevaccess t}_\agent \phi) \colonequals \{\jbox{\tprevaccess t}_\agent \phi\} \cup \Subf(\phi)$, $\phi$ is not a $\wprevious$-formula.
		\item[$\circ$] $\Subf(\jbox{\tprevaccess t}_\agent \wprevious \phi) \colonequals \{\jbox{\tprevaccess t }_\agent \wprevious \phi\} \cup \Subf(\jbox{t}_\agent \lsofar \phi)$.
		\vspace*{0.1cm}
		\item $\Subf(\jbox{\twprevright t}_\agent \phi) \colonequals \{\jbox{\twprevright t}_\agent \phi\} \cup \Subf(\jbox{t}_\agent \wprevious \phi)$.
		\vspace*{0.1cm}
		\item[$\circ$] $\Subf(\jbox{\tsprevright t}_\agent \phi) \colonequals \{\jbox{\tsprevright t}_\agent \phi\} \cup \Subf(\jbox{t}_\agent \sprevious \phi)$.
		\vspace*{0.1cm}
		\item $\Subf(\jbox{\tsprevleft t}_\agent \phi) \colonequals \{\jbox{\tsprevleft t}_\agent \phi\} \cup \Subf(\phi)$, $\phi$ is not a $\sprevious$-formula.
		\item $\Subf(\jbox{\tsprevleft t}_\agent \sprevious \phi) \colonequals \{\jbox{\tsprevleft t }_\agent \sprevious \phi\} \cup \Subf(\sprevious \jbox{t}_\agent  \phi)$.
		
	\end{itemize}
	Moreover, we extend the subformula relation $\Subf$ by transitivity.
\end{definition}

\subsection{Completeness for $\generalizeprinciple$}
\label{sec:Completeness for generalize}

As before let $\LPLTL^I(Ax)$ denote the result of adding axioms  from $Ax$ to $\LPLTL^I$. To keep the notation simple,  let ${\sf L}^{\sf gen}$ denote $\LPLTL^I (\{\generalizeprinciple\})$. In this section we aim to prove completeness of ${\sf L}^{\sf gen}$.

For a formula $\chi$, let 
\begin{gather*}
	B_\chi \colonequals \Subf(\chi) \cup  \Subf(\top \lsince \wprevious \bot)  \cup \Subf \{  \leventually \neg \jbox{t}_\agent \phi \ |\ \jbox{\tgeneralize t}_\agent \lalways \phi \in \Subf(\chi) \} 
	\\
	\cup\  \Subf \{  \lalways \jbox{\talwaysaccess t}_\agent \phi \ |\ \jbox{\talwaysaccess t}_\agent  \phi \in \Subf(\chi) \} \cup \Subf \{  \lsofar \jbox{\talwaysaccess_P t}_\agent \phi \ |\ \jbox{\talwaysaccess_P t}_\agent  \phi \in \Subf(\chi) \}
	\\
	\cup\ \Subf \{  \wprevious \jbox{\twprevright t}_\agent \phi \ |\ \jbox{\twprevright t}_\agent  \phi \in \Subf(\chi) \} \cup \Subf \{  \sprevious \jbox{\tsprevright t}_\agent \phi \ |\ \jbox{\tsprevright t}_\agent  \phi \in \Subf(\chi) \},
\end{gather*}
and
$$\Subf^+ (\chi) := B_\chi \cup \{ \neg \psi \ |\  \psi \in B_\chi \}.$$

Let $\MCS_\chi^\tgeneralize$ denote the set of all $\chi$-maximally $\mathsf{L}_\CS$-consistent subsets of $\Subf^+ (\chi)$.

For $\Gamma\in \MCS$, let\footnote{For simplicity we use the same symbol $\overline{\Gamma}$ as in Section \ref{sec:Maximal consistent sets}.} $$\overline{\Gamma} := \Gamma \cap \Subf^+ (\chi).$$ 

Note that all the results of Section \ref{sec:Maximal consistent sets} are valid if  $\Sub^+ (\chi)$ is replaced by $\Subf^+ (\chi)$, and $\vdash_\CS$ is replaced by $\vdash_{{\sf L}^{\sf gen}_\CS}$. Since the  proofs of the results of Section \ref{sec:Maximal consistent sets} have been given in details, we only outline the necessary changes here while omitting the proofs.

\begin{lemma}
	\[
	\MCS_\chi^\tgeneralize = \{\overline{\Gamma} \mid \Gamma\in\MCS \}.
	\]
\end{lemma}
\begin{proof}
Similar to the proof of Lemma \ref{lem:characterization of MCS-X} .\qed
\end{proof}

\begin{lemma}\label{lem: Facts about MCS-chi for generalize}
	Let  $\overline{\Gamma} \in \MCS_\chi^\tgeneralize$.
	
	\begin{enumerate}
			
		\item If $\phi \in B_\chi$ and  $\phi \not\in \overline{\Gamma}$, then $\neg \phi \in \overline{\Gamma}$.
		
		\item If $\phi \in \Subf^+(\chi)$ and  $\overline{\Gamma} \vdash_{{\sf L}^{\sf gen}_\CS} \phi$, then $\phi \in \overline{\Gamma}$.
		
		\item If $\psi \in \Subf^+(\chi)$, $\phi \in \overline{\Gamma}$ and $\vdash_{{\sf L}^{\sf gen}_\CS} \phi \rightarrow \psi$, then $\psi \in \overline{\Gamma}$.
	
	\end{enumerate}
\end{lemma}
\begin{proof}
	The proof of all items are standard. \qed 
\end{proof}

\begin{lemma}\label{lem:until-since-Subf}
    If either $\phi \luntil \psi \in \Subf^+(\chi)$ or $\phi \lsince \psi \in \Subf^+(\chi)$, then $\phi, \psi \in B_\chi$.
\end{lemma}
\begin{proof}
	We first show that if $\phi \luntil \psi \in \Subf^+(\chi)$, then $\phi, \psi \in B_\chi$. There are three cases:
	
	\begin{enumerate}
		\item $\phi \luntil \psi \in \Subf(\chi)$. Clearly $\phi, \psi \in \Subf(\chi)$, and hence $\phi, \psi \in B_\chi$.
		
		\item $\phi = \top$, $\psi = \neg \jbox{t}_i \sigma$ such that $\jbox{\tgeneralize t}_i \lalways \sigma \in \Subf(\chi)$. Then from $\neg \jbox{t}_i \sigma \in \Subf(\jbox{\tgeneralize t}_i \lalways \sigma)$, it follows that $\psi \in \Subf(\chi)$. Thus $\phi, \psi \in B_\chi$.
		
		\item $\phi \luntil \psi \in \Subf(\neg \jbox{t}_i \sigma)$ such that $\jbox{\tgeneralize t}_i \lalways \sigma \in \Subf(\chi)$. In this case $\phi \luntil \psi \in \Subf(\sigma)$. Since  $\sigma \in \Subf(\chi)$, we get $\phi \luntil \psi \in \Subf(\chi)$ and we reduce to case 1.
	\end{enumerate}

Now we show that if $\phi \lsince \psi \in \Subf^+(\chi)$, then $\phi, \psi \in B_\chi$. There are three cases:

\begin{enumerate}
	\item $\phi \lsince \psi \in \Subf(\chi)$. Clearly $\phi, \psi \in \Subf(\chi)$, and hence $\phi, \psi \in B_\chi$.
	
	\item $\phi = \top$, $\psi = \wprevious \bot$. Then $\phi, \psi \in B_\chi$ as desired.
	
	\item $\phi \lsince \psi \in \Subf(\top \luntil \neg \jbox{t}_i \sigma)$ such that $\jbox{\tgeneralize t}_i \lalways \sigma \in \Subf(\chi)$. In this case $\phi \lsince \psi \in \Subf(\neg \jbox{t}_i \sigma)$, and hence $\phi \lsince \psi \in \Subf(\jbox{t}_i \sigma)$. Since $\neg \jbox{t}_i \sigma \in \Subf(\jbox{\tgeneralize t}_i \lalways \sigma)$,  we get $\phi \lsince \psi \in \Subf(\chi)$ and we reduce to case 1. \qed
\end{enumerate}
\end{proof}

Using Lemma \ref{lem:until-since-Subf}, it is not difficult to show the following results (the proofs are similar to the proofs of Lemmas \ref{lem:until-sequence} and  \ref{lem:since-sequence} and thus are omitted here).

\begin{lemma}\label{lem:until-sequence for generalize}
	For every $\overline{\Gamma} \in \MCS_\chi^\tgeneralize$, if $\phi \luntil \psi \in \overline{\Gamma}$, then there exists a  $\phi \luntil \psi$-sequence starting with $\overline{\Gamma}$.
\end{lemma}

\begin{lemma}\label{lem:since-sequence for generalize}
	For every $\overline{\Gamma} \in \MCS_\chi^\tgeneralize$, if $\phi \lsince \psi \in \overline{\Gamma}$, then there exists a  $\phi \lsince \psi$-sequence ending with $\overline{\Gamma}$.
\end{lemma}

\begin{corollary}
	For every $\overline{\Gamma} \in \MCS_\chi$, there is an acceptable sequence containing $\overline{\Gamma}$. 
\end{corollary}

The following is an auxiliary lemma to be used in the proof of completeness.

\begin{lemma}\label{lem:evidence generalize condition}
Let $(\overline{\Gamma_0}, \overline{\Gamma_1}, \ldots)$ be an acceptable sequence of elements of $\MCS_\chi^\tgeneralize$. If $\jbox{\tgeneralize t}_i \lalways \phi \in \Subf(\chi)$ and $\jbox{t}_i \phi \in \Gamma_m$ for all $m \geq n$, then $ \jbox{\tgeneralize t}_i \lalways \phi \in \Gamma_n$.
\end{lemma}
\begin{proof}
Suppose towards a contradiction that $ \jbox{\tgeneralize t}_i \lalways \phi \not\in \Gamma_n$. Thus $\neg \jbox{\tgeneralize t}_i \lalways \phi \in \Gamma_n$, and then by $\generalizeprinciple$ we have $\neg \lalways \jbox{t}_i \phi \in \Gamma_n$. Hence $\leventually \neg \jbox{t}_i \phi \in \Gamma_n$. Note that $\leventually \neg \jbox{t}_i \phi$ is an abbreviation for $\top \luntil \neg \jbox{t}_i \phi$. On the other hand, from $\jbox{\tgeneralize t}_i \lalways \phi \in \Subf(\chi)$ it follows that $\top \luntil \neg \jbox{t}_i \phi \in \Subf^+ (\chi)$. Thus  $\top \luntil \neg \jbox{t}_i \phi \in \overline{\Gamma_n}$. Since $(\overline{\Gamma_0}, \overline{\Gamma_1}, \ldots)$ is an acceptable sequence, there exists $m \geq n$ such that $\neg \jbox{t}_i \phi \in \overline{\Gamma_m}$, and hence $\neg \jbox{t}_i \phi \in \Gamma_m$ which contradicts the hypothesis of the Lemma.\qed
\end{proof}

Given an ${\sf L}^{\sf gen}$-model $\M = (r,S,\evidence_1,\ldots,\evidence_\numberofagents, \valuation)$ for $\CS$ and a ternary relation $\B \subseteq S \times \Terms \times \Formulae$ and an agent $i$,
we define an operator 
\[
\Op{\B}_i: \powerset(S \times \Terms \times \Formulae)  \to  \powerset(S \times \Terms \times \Formulae)
\]
for $\CS$ by
\begin{equation*}
\begin{split}
\OpB(X) := \{ &(r(n),t,\phi) \ |\ \\
& (r(n),t,\phi) \in \B \ \lor \\
& \exists r, s (t=r+s \land ( (r(n),r,\phi) \in X \lor   (r(n),s,\phi) \in X )) \ \lor \\
& \exists r,s,\psi (t = r \cdot_\psi s \land (r(n),r,\psi \to \phi) \in X \land (r(n),s,\psi) \in X) \ \lor\\
& \exists r, \psi  (t= \tinspect r \land \phi = \jbox{r}_i \psi  \land  (r(n),r,\psi) \in X )  \ \lor\\
& \exists r,\psi ( t = \tgeneralize r \land \phi = \lalways \psi \land \forall m \geq n (r(m),r,\psi) \in X) 
\}
\end{split}
\end{equation*}

Obviously $\OpB$ is monotone, i.e.
\[
X \subseteq Y
\quad\text{implies}\quad
\OpB(X) \subseteq \OpB(Y).
\]
Therefore, $\OpB$ has a least fixed point, which we denote by $\evidence_i^\mathcal{B}$. That means $\evidence_i^\mathcal{B}$ is the least $X \subseteq  S \times \Terms \times \Formulae$ with
$X = \OpB(X)$.

\begin{definition}\label{def:canonical M-model for generalize}
	Let $(X_0, X_1, \ldots)$ be an acceptable sequence of elements of $\MCS_\chi^\tgeneralize$ for ${\sf L}^{\sf gen}_\CS$. The $\chi$-canonical  model~$\M = (r,S,\evidence_1,\ldots,\evidence_\numberofagents, \valuation)$ for $\CS$
	with respect to $(X_0, X_1, \ldots)$
	is  defined as follows:
	\begin{enumerate}
		
		\item $S \colonequals \{X_0, X_1, \ldots\}$.
		\item $r(n) := X_n$.\\
		Now we define relations $\B_\agent \subseteq  S \times \Terms \times \Formulae$ for each agent $\agent \in \Ag$ by
		\[
		(X_n,t,\phi) \in \B_\agent
		\quad\text{if{f}}\quad
		X_n \vdash_{{\sf L}_\CS^{\sf gen}} \jbox{t}_\agent \phi
		\] 

		\item $\evidence_\agent(X_n,t) \colonequals \{ \phi\ |\ (X_n, t, \phi) \in \evidence_\agent^{\mathcal{B}_\agent} \}$.
		\item $\valuation(X_n) \colonequals \Prop\cap X_n$.
		
	\end{enumerate}
	
\end{definition}

%
%
%
%
%
%

\begin{lemma}
	The $\chi$-canonical  model~$\M = (r,S,\evidence_1,\ldots,\evidence_\numberofagents, \valuation)$ for $\CS$ with respect to an acceptable sequence $(X_0, X_1, \ldots)$ is an ${\sf L}^{\sf gen}_\CS$-model.
\end{lemma}
\begin{proof}
	We only verify the condition $\generalizeevidence$ of Definition \ref{def: evidence functions M-models aditional principles}.
	Suppose that 
	$\phi \in \evidence_i(r(m),t)$ for all $m \geq n$. Thus $(r(m), t, \phi) \in \evidence_i^{\mathcal{B}_i}$  for all $m \geq n$.
	Since $\evidence_i^{\mathcal{B}_i}$ is a fixed point of $\Op{\mathcal{B}_i}_i$, we immediately get
	$(r(n), \tgeneralize t, \lalways \phi) \in \evidence_i^{\mathcal{B}_i}$. Hence $\lalways \phi \in \evidence_i(r(n),\tgeneralize t)$, as desired. \qed
\end{proof}

\begin{lemma}\label{lem:fixed-point implies base}
	If $\jbox{t}_i \phi \in \Subf^+ (\chi)$ and $(r(n),t,\phi) \in \evidence_i^{\mathcal{B}_i}$, then $\jbox{t}_i \phi \in \Gamma_{n}$. 
\end{lemma}
\begin{proof}
	By induction on the build-up of $\evidence_i^{\B_i}$. We distinguish the following cases:
	
	\begin{enumerate}
		\item Base case. The case $(r(n),t,\phi) \in \mathcal{B}_i$ is trivial.
		
		\item $\exists r, s (t=r+s \land ( (r(n),r,\phi) \in \evidence_i^{\B_i} \lor   (r(n),s,\phi) \in \evidence_i^{\B_i} ))$. Since $\jbox{t}_i \phi \in \Subf^+ (\chi)$  we get $\jbox{r}_i \phi \in \Subf^+ (\chi)$ and $\jbox{s}_i \phi \in \Subf^+ (\chi)$. By I.H. we get $\jbox{r}_i \phi \in \Gamma_n$ or $\jbox{s}_i \phi \in \Gamma_n$. Then $\jbox{r+s}_i \phi \in \Gamma_n$, and thus $\jbox{t}_i \phi \in \Gamma_{n}$. The case where $t = r \cdot_\psi s$ or $t = !r$ is treated similarly.

		\item $\exists r,\psi ( t = \tgeneralize r \land \phi = \lalways \psi \land \forall m \geq n (r(m),r,\psi) \in \evidence_i^{\mathcal{B}_i})$. It is easy to show that from $\jbox{t}_i \phi \in \Subf^+ (\chi)$ it follows that $\jbox{t}_i \phi \in \Subf (\chi)$ and $\jbox{r}_i \psi \in \Subf^+ (\chi)$. By the induction hypothesis, for all $m \geq n$ we have $\jbox{r}_i \psi \in \Gamma_m$. By Lemma \ref{lem:evidence generalize condition}, we get $ \jbox{\tgeneralize r}_i \lalways \psi \in \Gamma_n$, and therefore  $\jbox{t}_i \phi \in \Gamma_{n}$ as desired. \qed
	\end{enumerate}
\end{proof}

\begin{lemma}[Truth Lemma]
	Let $\M = (r,S,\evidence_1,\ldots,\evidence_\numberofagents, \valuation)$ be the $\chi$-canonical model for $\CS$ with respect to an acceptable sequence $(X_0, X_1, \ldots)$. For every formula $\psi \in \Subf^+ (\chi)$, and every $n \in \N$ we have:
	\[ 
	(\M, r(n)) \models \psi 
	\quad\text{if{f}}\quad
	\psi \in r(n). 
	\]
\end{lemma}
\begin{proof}
	As usual, the proof is by induction on the structure of $\psi$. We show only the following case:
	\begin{itemize}
		\item $\psi = \jbox{t}_i \phi$.
		
		$(\Rightarrow)$ If $(\M, r(n)) \models \jbox{t}_i \phi$, then $(r(n),t,\phi) \in \evidence_i^{\mathcal{B}_i}$. Thus, by Lemma \ref{lem:fixed-point implies base},  $\jbox{t}_i \phi \in r(n)$.
		
		$(\Leftarrow)$ If $\jbox{t}_i \phi \in r(n)$, then $(r(n),t,\phi) \in \evidence_i^{\mathcal{B}_i}$. By $\refax$,  we have  $\phi \in \Gamma_n$ and by I.H.~we get $(\M, r(n)) \models \phi$. We conclude $(\M, r(n)) \models \jbox{t}_i \phi$.
		\qed
	\end{itemize}
\end{proof}

\begin{theorem}[Soundness and completeness]\label{thm:completeness M-models-LPLTL^I generalize}
	Let  $\CS$ be a constant specification for ${\sf L}^{\sf gen}$. Then we have
	 $ \vdash_{{\sf L}^{\sf gen}_\CS} \phi$ iff $\M\entails \phi$ for all ${\sf L}^{\sf gen}_\CS$-models $\M$.
\end{theorem}

\begin{theorem}[Soundness and completeness]\label{thm:completeness-generalize-interpreted systems}
	Let $Ax = \{ \generalizeprinciple, \nextleftshiftprinciple \}$. Then $\LPLTL^I(Ax)_\CS$ is sound and complete with respect to all interpreted systems of $\LPLTL^I$ satisfying $\generalizeevidence$, $\nextleftshiftevidence$, and $\nextleftshiftR$.
\end{theorem}
\begin{proof}
	We detail the proof for the soundness part. The proof of completeness is similar to the proof of Theorem \ref{thm:completeness M-models-LPLTL^I generalize}.
	
	Let $Ax = \{ \generalizeprinciple, \nextleftshiftprinciple \}$ and $\system = (\runs, S, R_1,\ldots,R_\numberofagents, \evidence_1,\ldots,\evidence_\numberofagents, \valuation)$ be an arbitrary interpreted system for $\LPLTL^I(Ax)$. For an arbitrary $r \in \runs$ and $n \in \N$, assume $(\system, r, n) \models \lalways \jbox{t}_i \phi$. Thus, $(\system, r, m) \models  \jbox{t}_i \phi$ for every $m \geq n$. Hence, $\phi \in \evidence_\agent (r(m), t)$ for every $m \geq n$. By $\generalizeevidence$, we get $\lalways \phi \in \evidence_\agent (r(n), \tgeneralize t)$. Now let $r(n) R_i r'(n')$ and $m' \geq n'$, for arbitrary $r' \in \runs$ and arbitrary $n', m' \in \N$. By $\nextleftshiftR$ we have $r(n + m' - n') R_i r'(m')$. On the other hand, from the assumption we have $(\system, r, n + m' - n') \models  \jbox{t}_i \phi$. Thus, $(\system, r', m') \models   \phi$. Since $m' \geq n'$ was chosen arbitrary we get $(\system, r', n') \models  \lalways \phi$, and since $r'(n')$ was chosen arbitrary, we get $(\system, r, n) \models  \jbox{\tgeneralize t}_i \lalways \phi$ as desired. \qed
	
\end{proof}

We close this section with remarking that it is quit possible to extend this completeness result to extensions of ${\sf L}^{\sf gen}$. For example consider the logic $\LPLTL^I (Ax)$ where $Ax = \{ \generalizeprinciple, \pastgeneralizeprinciple \}$. In order to prove completeness for \linebreak $\LPLTL^I (Ax)$, redefine the operator $\Phi^\mathcal{B}_i$ as follows:

\begin{equation*}
\begin{split}
\OpB(X) := \{ &(r(n),t,\phi) \ |\ \\
& (r(n),t,\phi) \in \B \ \lor \\
& \exists r, s (t=r+s \land ( (r(n),r,\phi) \in X \lor   (r(n),s,\phi) \in X )) \ \lor \\
& \exists r,s,\psi (t = r \cdot_\psi s \land (r(n),r,\psi \to \phi) \in X \land (r(n),s,\psi) \in X) \ \lor\\
& \exists r, \psi  (t= \tinspect r \land \phi = \jbox{r}_i \psi  \land  (r(n),r,\psi) \in X )  \ \lor\\
& \exists r,\psi ( t = \tgeneralize r \land \phi = \lalways \psi \land \forall m \geq n (r(m),r,\psi) \in X)  \ \lor\\
& \exists r,\psi ( t = \tgeneralize_P r \land \phi = \lsofar \psi \land \forall m \leq n (r(m),r,\psi) \in X) 
\}
\end{split}
\end{equation*}

The rest of the proof of soundness and completeness is similar to that of ${\sf L}^{\sf gen}$.


\section{Internalization}
\label{sec:Internalization}

\begin{definition}
	A justification logic $\mathsf{L}$ satisfies \emph{internalization} if for each formula $\phi$ with
	$
	 \vdash_\mathsf{L} \phi
	$
	and for each agent $\agent$, there exists a term $t$ with
	$
	\vdash_\mathsf{L} \jbox{t}_\agent \phi 
	$.
\end{definition}

$\LPLTLp$ satisfies a restricted form of internalization.

\begin{lemma}\label{lem: internalization LPLTLp}
	Let\/ $\CS$ be an axiomatically appropriate constant specification for $\LPLTLp$. For each formula~$\phi$ and each  $\agent$, if $\vdash_\CS \phi$, and $\mprule$ and $\iteratedconstnecrule$ are the only rules that are used in the derivation of $\phi$, then $\vdash_\CS  \jbox{t}_\agent \phi$  for some term $t$.
\end{lemma}
\begin{proof}
	We proceed by induction on the derivation of $\phi$.
	
	In case $\phi$ is an axiom, since $\CS$ is axiomatically appropriate, there is a constant $c$ with
	\[
	 \vdash_\CS \jbox{c}_\agent \phi.
	\]
	
	In case $\phi$ is derived by modus ponens from $\psi \limplies \phi$ and $\psi$, then, by the induction hypothesis, there are terms $s_1$ and $s_2$ such that $\jbox{s_1}_\agent (\psi \limplies \phi)$ and $\jbox{s_2}_\agent \psi$ are provable.
	Using $\appax$ and modus ponens, we obtain $\jbox{s_1 \tapp s_2}_\agent \phi$. 
	
	In case $\phi$ is $\jbox{c_{j_n}}_{i_n}\ldots\jbox{c_{j_1}}_{i_1} \psi$, derived using \iteratedconstnecrule, since $\CS$ is axiomatically appropriate, we can use $\iteratedconstnecrule$ again to obtain 
	$\jbox{c_{j_{n+1}}}_{\agent} \phi$ for some justification constant $c_{j_{n+1}}$. \qed
\end{proof}

Next we shall extend $\LPLTLp$ to obtain a justification logic with the internalization property. Although the following two formulas are provable in $\LPLTLp$, see Lemma \ref{lem:mix:1}, in order to get the internalization property we need to add them as axioms:

\begin{enumerate}
	\item $\lalways \phi \to \lnext \phi$ \hfill \mixaxone
	\item $\lsofar \phi \to \wprevious \phi$ \hfill \mixaxtwo
\end{enumerate}

Let  $\LPLTL^\mathsf{int}$ be the logic $\LPLTLp$ extended by the axioms $\generalizeprinciple$,\linebreak $\pastgeneralizeprinciple$, $\mixaxone$, and $\mixaxtwo$.

\begin{theorem}[Internalization]\label{thm:internalization LPLTL^int}
	Let\/ $\CS$ be an axiomatically appropriate constant specification for $\LPLTL^\mathsf{int}$. 
	The system $\LPLTL^\mathsf{int}_\CS$ enjoys internalization.
\end{theorem}
\begin{proof}
	Suppose that $\phi$ is provable in $\LPLTL^\mathsf{int}_\CS$. Let $i$ be an arbitrary agent. We have to show that $\jbox{t}_i \phi$ is provable in $\LPLTL^\mathsf{int}_\CS$, for some term $t$. We proceed by induction on the derivation of $\phi$. We only consider the following cases:
		
	In case $\phi$ is $\lalways \psi$, derived using $\alwaysnecrule$, then, by the induction hypothesis, there is a term~$s$ such that $\jbox{s}_\agent \psi$ is provable.
	Now, we can use $\alwaysnecrule$ in order to obtain $\lalways \jbox{s}_\agent \psi$ and then $\generalizeprinciple$ and modus ponens to get $
	\jbox{\tgeneralize s}_\agent \lalways \psi
	$.
	
	In case $\phi$ is $\lnext \psi$, derived using $\nextnecrule$, then, as above, we obtain $\jbox{\tgeneralize s}_\agent \lalways \psi$.
	Since $\CS$ is axiomatically appropriate, there is a constant~$c$ with $\jbox{c}_\agent (\lalways \psi \to \lnext \psi)$.
	Thus we finally conclude
	$
	\jbox{c\ \tapp \tgeneralize s}_\agent \lnext \psi
	$.

In case $\phi$ is $\lsofar \psi$, derived using $\sofarnecrule$, then, by the induction hypothesis, there is a term~$s$ such that $\jbox{s}_\agent \psi$ is provable.
Now, we can use $\sofarnecrule$ in order to obtain $\lsofar \jbox{s}_\agent \psi$ and then $\pastgeneralizeprinciple$ and modus ponens to get $
\jbox{\tgeneralize_P s}_\agent \lsofar \psi
$.

In case $\phi$ is $\wprevious \psi$, derived using $\prevnecrule$, then, as above, we obtain $\jbox{\tgeneralize_P s}_\agent \lsofar \psi$.
Since $\CS$ is axiomatically appropriate, there is a constant~$c$ with $\jbox{c}_\agent (\lsofar \psi \to \wprevious \psi)$.
Thus we finally conclude
$
\jbox{c\ \tapp \tgeneralize_P s}_\agent \wprevious \psi
$. \qed
\end{proof}

\begin{remark}
	It is worth noting that there are already some known temporal justification logics that satisfy internalization, although they are formalized using only future operators.  Bucheli in \cite{Bucheli15} show that, for axiomatically appropriate constant specifications, the logics $\LPLTL + \generalizeprinciple + \nextaccessprinciple$ and $\LPLTL + \generalizeprinciple + \alwaysaccessprinciple + \nextleftshiftprinciple$ satisfy internalization.\footnote{Note that the background logic used by Bucheli in \cite{Bucheli15} is different from $\LPLTL$.} In \cite{BucheliGhariStuder2017} the authors introduced another extension of $\LPLTL$, which was called $\LPLTL^\star$ there, that satisfies internalization.
\end{remark}

\begin{theorem}[Internalization]\label{thm:internalization}
 Let $\CS$ be an axiomatically appropriate constant specification for $\LPLTLp(Ax)$ where 
 \[\{ \generalizeprinciple, \pastgeneralizeprinciple, \nextaccessprinciple, \wprevaccessprinciple \} \subseteq Ax.\]
 Then $\LPLTLp(Ax)_\CS$ enjoys internalization.
\end{theorem}
\begin{proof}
	The proof is similar to the proof of Theorem \ref{thm:internalization LPLTL^int}. We only consider the following cases:
	
	In case $\phi$ is $\lnext \psi$, derived using $\nextnecrule$, then, as in the proof of Theorem \ref{thm:internalization LPLTL^int}, we obtain $\jbox{\tgeneralize s}_\agent \lalways \psi$. Then, by $\nextaccessprinciple$, we get $\jbox{\tnextaccess\tgeneralize s}_\agent \lnext \psi$.

	In case $\phi$ is $\wprevious \psi$, derived using $\prevnecrule$, then, as in the proof of Theorem \ref{thm:internalization LPLTL^int}, we obtain $\jbox{\tgeneralize_P s}_\agent \lsofar \psi$. Then, by $\wprevaccessprinciple$, we get $\jbox{\downarrow_P\tgeneralize_P s}_\agent \wprevious \psi$.
	\qed
\end{proof}

In Theorems \ref{thm:internalization LPLTL^int} and \ref{thm:internalization} we present two logics that satisfy internalization. We now prove that these two logics have the following relationship.

\begin{lemma}
	 Let $\CS$ be an axiomatically appropriate constant specification for $\LPLTLp(Ax)$ where $\{ \mixaxone, \mixaxtwo \} \subseteq Ax.$
	For every agent~$\agent$, formula~$\phi$ and term~$t$ there are terms~$s_1(t)$ and $s_2(t)$ such that
		\[
		\vdash_{\LPLTLp(Ax)_\CS} \jbox{t}_\agent \lalways \phi \limplies \jbox{s_1(t)}_\agent \lnext \phi, \quad \mbox{and}
		\]
		\[
		\vdash_{\LPLTLp(Ax)_\CS} \jbox{t}_\agent \lsofar \phi \limplies \jbox{s_1(t)}_\agent \wprevious \phi.
		\]
		Thus, versions of $\nextaccessprinciple$ and $\wprevaccessprinciple$ are derivable in $\LPLTLp(Ax)_\CS$.
\end{lemma}

\begin{proof}
	Since $\CS$ is axiomatically appropriate and $\mixaxone$ and $\mixaxtwo$ are axioms of $\LPLTL^\mathsf{int}$, there are justification constants $a$ and $b$ such that $\jbox{a}_i (\lalways \phi \to \lnext \phi) \in \CS$ and  $\jbox{b}_i (\lsofar \phi \to \wprevious \phi) \in \CS$. Thus 
	\[
	\vdash_{\LPLTL^\mathsf{int}_\CS} \jbox{t}_\agent \lalways \phi \limplies \jbox{a \cdot t}_\agent \lnext \phi, \quad \mbox{and}
	\]
	\[
	\vdash_{\LPLTL^\mathsf{int}_\CS} \jbox{t}_\agent \lsofar \phi \limplies \jbox{b \cdot t}_\agent \wprevious \phi
	\]
	Finally put $s_1(t) := a \cdot t$ and $s_2(t) := b \cdot t$.\qed
\end{proof}

Combining Theorems \ref{thm:internalization LPLTL^int}, \ref{thm:internalization} with the results of Section \ref{sec:Completeness for generalize}  we then can obtain temporal justification logics, based on $\LPLTL^I$, that satisfy both internalization and completeness. Note that, since $\mixaxone$ and $\mixaxtwo$ are true in all $\LPLTL^I$-models, the class of all models of $$\LPLTL^I (\{\generalizeprinciple, \pastgeneralizeprinciple, \mixaxone, \mixaxtwo\})$$ is the same as the class of all models of $$\LPLTL^I(\{\generalizeprinciple, \pastgeneralizeprinciple\}).$$

\begin{theorem}[Completeness and Internalization]\label{thm:completeness and internalization}
	Let ${\sf L}$ be the logic $\LPLTL^I$ extended by either of the following set of axioms:
	\begin{enumerate}
		\item $\{\generalizeprinciple, \pastgeneralizeprinciple, \mixaxone, \mixaxtwo\}$, or
		
		\item $\{ \generalizeprinciple, \pastgeneralizeprinciple, \nextaccessprinciple, \wprevaccessprinciple \}$.
	\end{enumerate}
	Let $\CS$ be an axiomatically appropriate constant specification for ${\sf L}$.	Then ${\sf L}_\CS$ enjoys internalization and is sound and complete with respect to ${\sf L}_\CS$-models.
\end{theorem}
\begin{proof}
	Follows from Theorems \ref{thm:internalization LPLTL^int}, \ref{thm:internalization}, \ref{thm:completeness M-models-LPLTL^I generalize}. \qed
\end{proof}
\section{No forgetting and no learning}
\label{sec:No forgetting and no learning}

\textit{No forgetting} (or \textit{perfect recall}) and \textit{no learning} are two well known properties of systems that can be expressed in the language of logics of knowledge and time. It seems that the axioms $\alwaysaccessprinciple$ and $\pastaccessprinciple$ correspond respectively to the notions of no forgetting and no learning on justifications. Let’s make this precise. 

A formula $\phi$ is said to be \textit{stable with respect to the future} if once it is true it remains true, i.e. $\vdash \phi \to \lalways \phi$. In the framework of logics of knowledge and time, it is known that if a logic contains the axiom $\lknows_\agent \lalways \phi \limplies \lalways \lknows_\agent \phi$, then for every formula $\phi$ which is  stable with respect to the future it can be shown that $\vdash \lknows_\agent \phi \to \lalways \lknows_\agent \phi$, i.e. if $\phi$ is known at some point then it remains known at all points in the future (see \cite{FHMV95}). Likewise, we show that logics that contain axiom $\alwaysaccessprinciple$, i.e. $\jbox{t}_\agent \lalways \phi \limplies \lalways \jbox{\talwaysaccess t}_\agent \phi$, have a similar property.

\begin{theorem}\label{thm: alwaysaccess principle-no forgetting}
	Let $Ax \supseteq \{\alwaysaccessprinciple\}$ and let $\mathsf{L} = \LPLTLp(Ax)_\CS$ be a justification logic that satisfies internalization. If 
	\[
	\vdash_\mathsf{L} \phi \to \lalways \phi,
	\]
	then for every term $t$ there is a term $s(t)$ such that
	\[
	\vdash_\mathsf{L} \jbox{t}_\agent \phi \to \lalways \jbox{s(t)}_\agent \phi.
	\]
\end{theorem}
\begin{proof}
	Suppose that $\phi \to \lalways \phi$ is provable in $\LPLTLp(Ax)_\CS$, where $Ax \supseteq \{\alwaysaccessprinciple\}$. Thus, by the internalization property of $\LPLTLp(Ax)_\CS$, we get $\jbox{r}_\agent (\phi \to \lalways \phi)$  for some term $r$. Hence, for every term $t$, $\jbox{t}_\agent \phi \to \jbox{r \cdot_\phi t}_\agent \lalways \phi$, and therefore by axiom $\alwaysaccessprinciple$ we get $\jbox{t}_\agent \phi \to \lalways  \jbox{\talwaysaccess (r \cdot_\phi t)}_\agent \phi$. Thus, for every term $t$ it is enough to put $s \colonequals \talwaysaccess (r \cdot_\phi t)$. \qed 
\end{proof}

Using past time operators, a similar argument can be done for no learning. A formula $\phi$ is said to be \textit{stable with respect to the past} if once it is true it has always been true, i.e. $\vdash \phi \to \lsofar \phi$. Using axiom $\lknows_\agent \lsofar \phi \limplies \lsofar \lknows_\agent \phi$, it is easy to show that for every formula $\phi$ which is  stable with respect to the past we have $\vdash \lknows_\agent \phi \to \lsofar \lknows_\agent \phi$, i.e. if $\phi$ is known at some point then it has always been known at all points in the past. Note that, since $\lsofar \psi \limplies \lonce \psi$ is a valid formula for every $\psi$, $\lknows_\agent \phi \to \lsofar \lknows_\agent \phi$  in turn entails $\lknows_\agent \phi \to \lonce \lknows_\agent \phi$, i.e. if $\phi$ is known at some point then it was known at some point in the past. We show that logics that contain axiom $\pastaccessprinciple$, i.e. $\jbox{t}_\agent \lsofar \phi \limplies \lsofar \jbox{\talwaysaccess_P t}_\agent \phi$, have a similar property.

\begin{theorem}\label{thm: pastaccess principle-no learning}
	Let $Ax \supseteq \{ \pastaccessprinciple \}$ and let $\mathsf{L} = \LPLTLp(Ax)_\CS$ be a justification logic that satisfies internalization. If 
	\[
	\vdash_\mathsf{L} \phi \to \lsofar \phi,
	\]
	then for every term $t$ there is a term $s(t)$ such that
	\[
	\vdash_\mathsf{L} \jbox{t}_\agent \phi \to \lsofar \jbox{s(t)}_\agent \phi.
	\]
\end{theorem}
\begin{proof}
	Suppose that $\phi \to \lsofar \phi$ is provable in $\LPLTLp(Ax)_\CS$, where $Ax \supseteq \{ \pastaccessprinciple \}$. Thus, by the internalization property of $\LPLTLp(Ax)_\CS$, we get $\jbox{r}_\agent (\phi \to \lsofar \phi)$  for some term $r$. Hence, for every term $t$, $\jbox{t}_\agent \phi \to \jbox{r \cdot_\phi t}_\agent \lsofar \phi$, and therefore by axiom $\pastaccessprinciple$ we get $\jbox{t}_\agent \phi \to \lsofar \jbox{\talwaysaccess_P (r \cdot_\phi t)}_\agent \phi$. Thus, for every term $t$ it is enough to put $s \colonequals \talwaysaccess_P (r \cdot_\phi t)$. \qed  
\end{proof}

In the framework of logics of knowledge and time, it is known that the following principles would characterize systems with \textit{no forgetting} $\nfax$ and \textit{no learning} $\nlax$ respectively (cf. \cite{FHMV95,HvdMV04}):
\begin{itemize}
	\item $\lknows_\agent \phi \lsince \lknows_\agent \psi \limplies \lknows_\agent(\lknows_\agent  \phi \lsince \lknows_\agent \psi)$ \hfill \nfax
	
	\item $\lknows_\agent \phi \luntil \lknows_\agent \psi \limplies \lknows_\agent(\lknows_\agent  \phi \luntil \lknows_\agent \psi)$ \hfill \nlax
\end{itemize}
Now let us consider the justification counterparts of the above axioms. The following principles could be considered as justification counterparts of $\nfax$ and $\nlax$ respectively
\begin{itemize}
	\item $\jbox{t}_\agent \phi \lsince \jbox{s}_\agent \psi \limplies \jbox{nf(t,s)}_\agent(\jbox{t}_\agent \phi \lsince \jbox{s}_\agent \psi)$ \hfill \jnfax
	
	\item $\jbox{t}_\agent \phi \luntil \jbox{s}_\agent \psi \limplies \jbox{nl(t,s)}_\agent(\jbox{t}_\agent \phi \luntil \jbox{s}_\agent \psi)$ \hfill \jnlax
\end{itemize}
where $nf$ and $nl$ are two binary new term operators.

Now we give a semantics for the logic $\LPLTLp \{ \jnfax, \jnlax \}$ similar to the semantics of Section \ref{sec:Mkrtychev models aditional principles}. Given a constant specification $\CS$ for $\LPLTLp \{ \jnfax, \jnlax \}$, an $\LPLTLp \{ \jnfax, \jnlax \}_\CS$-model is defined in the same manner as $\LPLTLp$-models (see Definition \ref{def:M-models LPLTLp}) with the following additional conditions $\jnfevidence$ and $\jnlevidence$ on evidence functions:

\begin{itemize}
	\item If there is some $m$ with $n \geq m \geq 0$ such that  $\psi \in \evidence_\agent (r(n-m),s)$ and for all $k$ with $0 \leq k < m$ we have $\phi \in \evidence_\agent (r(n-k),t)$, then $\jbox{t}_\agent \phi \lsince \jbox{s}_\agent \psi \in \evidence_\agent(r(n), nf(t,s))$. \hfill $\jnfevidence$
	
	\item If there is some  $m \geq 0$ such that  $\psi \in \evidence_\agent(r(n+m),s)$  and $\phi \in \evidence_\agent(r(n+k),t)$ for all $k$ with $0 \leq k < m$, then $\jbox{t}_\agent \phi \luntil \jbox{s}_\agent \psi \in \evidence_\agent(r(n),nl(t,s))$. \hfill $\jnlevidence$
\end{itemize} 

In order to prove the completeness of $\LPLTLp \{ \jnfax, \jnlax \}$, it is enough to add the following closure conditions to  conditions (1)--(4) of Definition  \ref{def: evidence functions M-models aditional principles}:
\begin{itemize}
	\item If there is some $m$ with $n \geq m \geq 0$ such that  $(r(n-m),s,\psi) \in E$ and for all $k$ with $0 \leq k < m$ we have $(r(n-k),t,\phi) \in E$, then $(r(n), nf(t,s),\jbox{t}_\agent \phi \lsince \jbox{s}_\agent \psi) \in \OpB(E)$. \hfill (cl-nf)
	
	\item If there is some  $m \geq 0$ such that  $(r(n+m),s,\psi) \in E$  and $(r(n+k),t,\phi) \in E$ for all $k$ with $0 \leq k < m$, then $(r(n),nl(t,s),\jbox{t}_\agent \phi \luntil \jbox{s}_\agent \psi) \in \OpB(E)$. \hfill (cl-nl)
\end{itemize}
Now soundness and completeness of $\LPLTLp \{ \jnfax, \jnlax \}$ is proved similar to that of $\LPLTLp$ in Section \ref{sec:Connecting principles}.

\begin{theorem}[Soundness and completeness]
	Let ${\sf L} = \LPLTLp \{ \jnfax, \jnlax \}$. For each formula $\chi$ and finite set of formulas $T$, we have $T \Vdash_{{\sf L}_\CS}  \chi$ if{f} $T \vdash_{{\sf L}_\CS} \chi$.
\end{theorem}
\begin{proof}
	Straightforward.  \qed
\end{proof}

Theorems \ref{thm: alwaysaccess principle-no forgetting} and \ref{thm: pastaccess principle-no learning} show the relationship between the axioms $\alwaysaccessprinciple$ and $\pastaccessprinciple$ with the notions of no forgetting and no learning, respectively. However, the relationship between the axioms $\jnfax$ and $\jnlax$ with the notions of no forgetting and no learning is not clear yet, except that they are justification counterparts of the axioms $\nfax$ and $\nlax$ respectively. We leave the study of this issue for future work.

\section{Reasoning takes time}
\label{sec:JLTL}

In this section we explore more interactions between justification and time. Let us start with the axiom  $\appax$:
\[
\jbox{t}_\agent (\phi \limplies \psi) \limplies (\jbox{s}_\agent \phi \limplies \jbox{t \tapp s}_\agent \psi).
\]
This axiom says that if  agent $i$ knows $\phi \limplies \psi$ for reason $t$ and she knows $\phi$ for reason $s$, then  \textit{at the same time} she knows $\psi$ for reason $t \tapp s$. So the agent applies the rule Modus Ponens $\mprule$ in her reasoning, but this step of reasoning takes no time. Thus, at a given moment of time the agent knows all consequences of her knowledge. This is related to the \textit{Logical Omniscience Problem}. This would be implausible if we expect that reasoning takes time. The same argument can be applied to the axioms $\sumax$ and $\posintax$. 

In \cite{BucheliGhariStuder2017} the following principles have been suggested to formalize the idea that reasoning with justifications takes time:
\begin{gather*}
	\jbox{t}_\agent (\phi \limplies \psi) \limplies (\jbox{s}_\agent \phi \limplies \lnext\jbox{t \tapp s}_\agent \psi), \\
	\jbox{t}_\agent \phi \vee \jbox{s}_\agent \phi \limplies  \lnext \jbox{t + s}_\agent \phi, \\
	\jbox{t}_\agent \phi \limplies \lnext \jbox{\tinspect t}_\agent \jbox{t}_\agent \phi\,.
\end{gather*} 

At first sight the above principles seem to be impeccable, but it is not difficult to show that they have the following implausible consequences:
\begin{gather}
	\jbox{t}_\agent \phi \rightarrow \lnext \phi, \label{eq: implausible consequences JLTL 1}
	\\
	\jbox{t}_\agent \phi \rightarrow \lnext\jbox{t}_\agent \phi,
	\\
	\jbox{t}_\agent \phi \rightarrow \lalways \jbox{t}_\agent \phi,
	\\
	\jbox{t}_\agent \phi \rightarrow \lalways \phi. \label{eq: implausible consequences JLTL 4}
\end{gather} 
where $t \in \Terms$ and $\phi \in \Formulae$ are arbitrary.

In the following we study another variant of the above principles.\footnote{Thanks to Thomas Studer for suggesting me these axioms.} The logic $\JLTL$ is defined similar to $\LPLTLp$ with the difference that axioms of the justification part are replaced by the following axioms

\begin{enumerate}
	\item $\jbox{t}_\agent (\phi \limplies \psi) \limplies (\jbox{s}_\agent \phi \limplies \lnext\jbox{t \tapp s}_\agent \sprevious\psi)$ \hfill \fpappax
	\item $\jbox{t}_\agent \phi \rightarrow \lnext \jbox{t + s}_\agent \sprevious\phi$, \quad $\jbox{s}_\agent \phi \limplies \lnext \jbox{t + s}_\agent \sprevious\phi$ \hfill \fpsumax
	\item $\jbox{t}_\agent \phi \limplies \phi$ \hfill \refax
	\item $\jbox{t}_\agent \phi \limplies \lnext \jbox{\tinspect t}_\agent \sprevious \jbox{t}_\agent \phi$ \hfill \fpposintax\footnote{The prefix ${\sf FP}$ in the name of these axioms comes from the first letters of `{\sf F}uture' and `{\sf P}ast'.}
\end{enumerate}

The axiom $\fpappax$ says that if  agent $i$ knows $\phi \limplies \psi$ for reason $t$ and she knows $\phi$ for reason $s$, then tomorrow she will know that  $\psi$ was the case yesterday for reason $t \tapp s$. So the agent applies the rule Modus Ponens $\mprule$ in her reasoning, and here this step of reasoning takes time. Thus $\JLTL$-agents avoid the logical omniscience problem.


Next we present a semantics for $\JLTL$ similar to $\LPLTLp$-models given in Section \ref{sec:Mkrtychev models aditional principles}.

\begin{definition}\label{def: LTLJ-models}
	An  $\JLTL_\CS$-model is a tuple $\M = (r, S, \evidence_1\ldots,\evidence_\numberofagents, \valuation)$
	where
	\begin{enumerate}
		\item $S$ is a non-empty set of states;
		\item $r : \N \to S$ is a run on $S$;
		\item $\evidence_\agent$ is an $\JLTL_\CS$-evidence function for each agent~$\agent \in \Ag$;
		\item $\valuation: S \to \mathcal{P}(\Prop)$ is a valuation.
	\end{enumerate}	
	$\JLTL_\CS$-evidence functions should
	satisfy the following conditions.
	For all $n \in \N$, all terms $s,t \in \Terms$ and all formulas $\phi,\psi \in \Formulae$: 
	\begin{enumerate}
		\item 
		if $\jbox{t}_\agent \phi \in \CS$, then $\phi \in \evidence_\agent(r(n),t)$, \hfill (constant specification)
		
		\item 
		if $\phi \limplies \psi \in \evidence_\agent(r(n),t)$ and $\phi \in \evidence_\agent(r(n),s)$, then $\sprevious\psi \in \evidence_\agent(r(n+1), t \tapp s)$, \\ \text{} \hfill {\sf (FP-application)}
		\item 
		if $\phi \in \evidence_\agent(r(n),s) \cup \evidence_\agent(r(n),t)$, then $\sprevious \phi \in \evidence_\agent(r(n+1),s + t)$, \hfill {\sf (FP-sum)}
		\item 
		if $\phi \in \evidence_\agent(r(n),t)$, then $\sprevious\jbox{t}_\agent \phi \in \evidence_\agent(r(n+1),\tinspect t)$. \hfill {\sf (FP-positive introspection)}
	\end{enumerate}
\end{definition}
Given an $\JLTL_\CS$-model $\M$, the truth of a formula  in $\M$ is defined in the same manner as in Definition \ref{def:M-models LPLTLp}. The proof of completeness is similar to the proof of Theorem \ref{thm:completeness LPLTLp M-models} by constructing a canonical model. Note that in order to prove the completeness, conditions (2)--(4) of Definition \ref{def:fixed point operator} should be replaced by the following closure conditions:
\begin{enumerate}
	
	\item 
	If $(r(n),t,\psi \limplies \phi) \in E$ and $(r(n),s,\psi) \in E$, then $(r(n+1),t \cdot_\psi s,\sprevious \phi) \in \OpB(E)$. \\ \text{} \hfill (FP-cl-application)
	
	\item 
	If $(r(n),t,\phi) \in E$ or $(r(n),s,\phi) \in E$, then $(r(n+1),s+t,\sprevious \phi) \in \OpB(E)$. \\ \text{} \hfill (FP-cl-sum)
	
	\item 
	If $(r(n),t,\phi) \in E$, then $(r(n+1),!t,\sprevious \jbox{t}_\agent \phi) \in \OpB(E)$. \\ \text{} \hfill (FP-cl-positive-introspection)
\end{enumerate}
Now soundness and completeness of $\JLTL_\CS$ is proved similar to that of $\LPLTLp$ in Section \ref{sec:Connecting principles}.

\begin{theorem}[Soundness and completeness]
	Let $\CS$ be a constant specification for $\JLTL$. Then $\vdash_{\JLTL_\CS} \phi$ if{f} $\M\entails \phi$ for all $\JLTL_\CS$-models $\M$.
\end{theorem}
\begin{proof}
	Soundness is straightforward. The proof of completeness is similar to the proof of Theorem \ref{thm:completeness LPLTLp M-models} by constructing a canonical model.  Truth Lemma can be proved as before. The only new part is to show that any $\chi$-canonical model for $\CS$ with respect to an acceptable sequence $(\overline{\Gamma_0}, \overline{\Gamma_1}, \ldots)$ for $\JLTL_\CS$ is an $\JLTL_\CS$-model.  This is left to the reader. \qed	 
\end{proof}

Given a set $Ax$ of connecting principles from Section \ref{sec:Axioms}, by $\JLTL(Ax)$ we denote the result of adding axioms from $Ax$ to $\JLTL$. The above completeness result can  be easily extended to $\JLTL(Ax)$ as well.

It is not difficult to show that none of the formulas \eqref{eq: implausible consequences JLTL 1}-\eqref{eq: implausible consequences JLTL 4} are valid in $\JLTL$. For example, consider the following instance of \eqref{eq: implausible consequences JLTL 1}-\eqref{eq: implausible consequences JLTL 4}:
\begin{equation}\label{eq: implausible consequences JLTL-instances}
	\jbox{x}_\agent P \rightarrow \lnext P, \quad
	\jbox{x}_\agent P \rightarrow \lnext\jbox{x}_\agent P, \quad
	\jbox{x}_\agent P \rightarrow \lalways \jbox{x}_\agent P, \quad
	\jbox{x}_\agent P \rightarrow \lalways P,
\end{equation}
where $x \in \VTerms$ and $P \in \Prop$. Let $\M = (r, S, \evidence_1\ldots,\evidence_\numberofagents, \valuation)$ be defined as follows:
\begin{itemize}
	\item $S = \{ w, v\}$.
	
	\item $r(0) = w$ and $r(n) = v$ for all $n \geq 1$.
	
	\item $P \in \valuation(w)$ and $P \not \in \valuation(v)$.
	
	\item $\evidence_\agent (r(n), t) = \{ \phi \mid (r(n), t, \phi) \in \evidence_\agent^\mathcal{B} \}$, where $\mathcal{B} = \{ (r(0), x, P) \}$ and $\evidence_\agent^\mathcal{B}$ is the least fixed point of $\Phi_\agent^\mathcal{B}$.\footnote{Note that the closure conditions (FP-cl-application), (FP-cl-sum), and (FP-cl-positive-introspection) are used in the definition of $\Phi_\agent^\mathcal{B}$.} 
\end{itemize}
Now it is obvious that $\M$ is an $\JLTL_\emptyset$-model, and further none of the formulas in \eqref{eq: implausible consequences JLTL-instances} are true in $\M$ at state $r(0)$. Thus, none of the formulas in \eqref{eq: implausible consequences JLTL-instances} are valid in $\JLTL_\emptyset$.
%


\begin{remark}
	Note that since $\phi \limplies \leventually \phi$ is provable in $\LTL$ the following formulas trivially follows  in $\LPLTL$ from the axioms $\appax$, $\sumax$, and \linebreak $\posintax$:
	\begin{gather*}
	\jbox{t}_\agent (\phi \limplies \psi) \limplies (\jbox{s}_\agent \phi \limplies \leventually \jbox{t \tapp s}_\agent \psi), \\
	\jbox{t}_\agent \phi \vee \jbox{s}_\agent \phi \limplies  \leventually \jbox{t + s}_\agent \phi, \\
	\jbox{t}_\agent \phi \limplies \leventually \jbox{\tinspect t}_\agent \jbox{t}_\agent \phi\,.
	\end{gather*} 
	%
	 A more realistic set of axioms which do not suffer from the logical omniscience problem can be formulated as follows:
	\begin{gather*}
	\jbox{t}_\agent (\phi \limplies \psi) \limplies (\jbox{s}_\agent \phi \limplies \langle F \rangle\, \jbox{t \tapp s}_\agent \psi), \\
	\jbox{t}_\agent \phi \vee \jbox{s}_\agent \phi \limplies  \langle F \rangle\, \jbox{t + s}_\agent \phi, \\
	\jbox{t}_\agent \phi \limplies \langle F \rangle\, \jbox{\tinspect t}_\agent \jbox{t}_\agent \phi\,.
	\end{gather*} 
	where $\langle F \rangle\, \phi := \neg \phi \wedge \leventually \phi$. We leave the proof of completeness to possible future work.
\end{remark}

\bibliography{library}
\end{document}